\documentclass[a4paper,UKenglish,cleveref, autoref, thm-restate,authorcolumns]{lipics-v2021}

\usepackage{graphicx} 

\usepackage{float}
\usepackage{tikz}
\usepackage{amsmath,amsthm,amssymb,amsfonts}
\usepackage{xcolor}
\usepackage{multirow}
\usepackage{subcaption}
\usepackage{array}

\usepackage[table]{xcolor}

\newcolumntype{M}[1]{>{\centering\arraybackslash}m{#1}}
\newcommand{\floor}[1]{\left\lfloor #1 \right\rfloor}

\newcommand{\N}{\mathbb{N}}
\newcommand{\Z}{\mathbb{Z}}

\newcommand{\Set}[1]{\{#1\}}

\newcommand{\crn}[1]{\mathcal{#1}}
\newcommand{\species}{\Lambda}
\newcommand{\reactions}{\Gamma}
\newcommand{\reactionsS}{\reactions_\states}
\newcommand{\reaction}{\mathcal{\gamma}}
\newcommand{\size}[1]{\lvert #1 \rvert}
\newcommand{\config}[1]{\overrightarrow{#1}}
\newcommand{\single}[1]{\vec{#1}}
\newcommand{\reactants}{\config{R}}
\newcommand{\products}{\config{P}}

\newcommand{\emptyConfig}{\config{0}}

\newcommand{\vas}[1]{\mathcal{#1}}
\newcommand{\transitions}{\mathcal{T}}

\newcommand{\states}{\mathcal{Q}}
\newcommand{\stepsto}{\mathop{\rightarrow}\limits}
\newcommand{\reaches}{\mathop{\leadsto}\limits}
\newcommand{\rxn}{\mathop{\longrightarrow}\limits}

\newcommand{\para}[1]{\noindent \textbf{#1}}

\title{Reachability with Restricted Reactions in Inhibitory Chemical Reaction Networks} 

\author{Divya Bajaj}{University of Texas Rio Grande Valley}{divya.bajaj@utrgv.edu}{}{}

\author{Bin Fu}{University of Texas Rio Grande Valley}{bin.fu@utrgv.edu}{}{}

\author{Ryan Knobel}{University of Texas Rio Grande Valley}{ryan.knobel01@utrgv.edu}{}{}

\author{Austin Luchsinger}{University of Texas Rio Grande Valley}{austin.luchsinger@utrgv.edu}{}{}

\author{Aiden Massie}{University of Texas Rio Grande Valley}{aiden.massie01@utrgv.edu}{}{}

\author{Pablo Santos}{University of Texas Rio Grande Valley}{pablo.santos01@utrgv.edu}{}{}

\author{Ramiro Santos}{University of Texas Rio Grande Valley}{ramiro.santos01@utrgv.edu}{}{}

\author{Robert Schweller}{University of Texas Rio Grande Valley}{robert.schweller@utrgv.edu}{}{}

\author{Evan Tomai}{University of Texas Rio Grande Valley}{evan.tomai@utrgv.edu}{}{}

\author{Tim Wylie}{University of Texas Rio Grande Valley}{timothy.wylie@utrgv.edu}{}{}

\authorrunning{D.~Bajaj et al.} 

\Copyright{Divya Bajaj \and Bin Fu \and Ryan Knobel \and Austin Luchsinger \and Aiden Massie \and Adrian Salinas \and Pablo Santos \and Ramiro Santos \and Robert Schweller \and Evan Tomai \and Tim Wylie}



\category{}

\relatedversion{Full Version}{}

\funding{This research was supported in part by National Science Foundation Grant CCF-2329918.}

\acknowledgements{}

\nolinenumbers 

\EventEditors{Pierre Fraigniaud}
\EventNoEds{1}
\EventLongTitle{20th Scandinavian Symposium on Algorithm Theory (SWAT 2026)}
\EventShortTitle{SWAT 2026}
\EventAcronym{SWAT}
\EventYear{2026}
\EventDate{June 17--19, 2026}
\EventLocation{Copenhagen, Denmark}
\EventLogo{}
\SeriesVolume{370}
\ArticleNo{5}

\begin{document}

\maketitle

\begin{abstract}
Chemical Reaction Networks (CRNs) are a well-established model of distributed computing characterized by quantities of molecular species that can transform or change through applications of reactions. A fundamental problem in CRNs is the reachability problem, which asks if an initial configuration of species can transition to a target configuration through an applicable sequence of reactions. It is well-known that the reachability problem in \emph{general} CRNs was recently proven to be Ackermann-complete. However, if the CRN's reactions are restricted in both power, such as only deleting species (deletion-only rules) or consuming and producing an equal number of species (volume-preserving rules), \emph{and} size (unimolecular or bimolecular rules), then reachability falls below Ackermann-completeness, and is even solvable in polynomial time for deletion-only systems.

In this paper, we investigate reachability under this set of restricted unimolecular and bimolecular reactions, but in the Priority-Inhibitory CRN and Inhibitory CRN models. These models extend a traditional CRN by allowing some reactions to be inhibited from firing in a configuration if certain species are present; the exact inhibition behavior varies between the models. We first show that reachability with Priority iCRNs mostly remains in P for deletion-only systems, but becomes NP-complete for one case. We then show that reachability with deletion-only reactions for iCRNs is mostly NP-complete, and PSPACE-complete even for $(1,1)$-size (general) reactions. We also provide FPT algorithms for solving most of the reachability problems for the iCRN model. Finally, we show reachability for CRNs with states is already NP-hard for the simplest deletion-only systems, and is PSPACE-complete even for (general) $(1,1)$-size reactions.
\end{abstract}

\newpage

\section{Introduction}

Chemical Reaction Networks (CRNs) \cite{Aris:1965:ARMA,Aris:1968:ARMA} are an abstract model of molecular computation in which molecular \emph{species} can evolve through applications of \emph{reactions}. Although initially motivated by modeling natural chemical interactions, CRNs have been recently explored for developing artificial chemical systems~\cite{NagipoguReif2025NeuralCRNs}, computations of semilinear functions \cite{chen2014deterministic}, digital logic \cite{cardelli:2018:chemical, hjelmfelt:1991:neural, jiang:2013:digital}, and neural networks \cite{hjelmfelt:1991:neural}, along with successful implementations in DNA Strand Displacement (DSD) systems \cite{soloveichik:2010:dna}. 

CRNs are computationally equivalent to two classical models of infinite-state systems: Petri-nets \cite{petri1966communication} and Vector Addition Systems (VAS) \cite{karp1967parallel,karp1969parallel}. This equivalence connects molecular computation with decades of complexity-theoretic results concerning reachability \cite{cook2009programmability, hack1976decidability}. In all three models, the \emph{reachability problem} asks whether a target configuration $\config{C_t}$ can be reached from a given initial configuration $\config{C_s}$ via a sequence of rule applications. Although certain characteristics of reachability in general CRNs (and equivalently, Petri-nets and VAS) were well-known for decades, such as EXPSPACE-hardness \cite{lipton1976reachability} and decidability \cite{mayr1981algorithm}, only recently has the problem been proven to be Ackermann-complete \cite{czerwinski2022reachability, leroux2022reachability}.

A crucial limitation of the classical models is the absence of \emph{zero-checking}. To address this, a natural and well-studied extension across all three models is to allow \emph{inhibition}, which allows certain reactions to be disabled by the presence of designated species. Petri-nets with inhibitor arcs \cite{hack1976petri,REINHARDT2008PVAS}, Priority Vector Addition Systems \cite{guttenberg2024flattability}, and Inhibitory Chemical Reaction Networks \cite{calabrese2024inhibitory} all explore this idea. Because inhibition enables zero-checking, these inhibitory extensions gain substantial power, achieving Turing Universality  \cite{calabrese2024inhibitory,peterson1981petri}. However, due to the Turing-universal power of zero-checking, the reachability problem becomes undecidable in general for such systems with inhibition. For this reason, and to better understand the boundaries between decidable/tractable and undecidable/intractable, we limit our study to restricted classes of CRNs with inhibition that use only deletion-only or volume-preserving rules.

For CRNs, \emph{deletion-only} (void) reactions never create any new species in the system (e.g., $\lambda_1+\lambda_2\rightarrow \emptyset$, or $\lambda_1+\lambda_2\rightarrow \lambda_1$), and \emph{volume-preserving} reactions contain an equal number of reactants and products, ensuring that the system's volume never changes (e.g., $\lambda_1+\lambda_2\rightarrow\lambda_3+\lambda_4$).
These restrictions are well-motivated: unlike general CRN systems, which can exhibit complex and error-prone behaviors that are challenging to implement \cite{wang:2018:effective}, deletion-only and volume-preserving systems are potentially simpler to realize and analyze \cite{AFG:2025:RRC, Fu:2025:DNA}.
Moreover, in the absence of inhibition, such restrictions significantly reduce computational power: deletion-only systems cannot compute even simple functions \cite{anderson2024steps}, reachability with bimolecular void reactions is solvable in polynomial time \cite{AFG:2025:RRC, Fu:2025:DNA}, and reachability with $(1,1)$-size reactions is solvable in polynomial time  \cite{AFG:2025:RRC}. 
These observations raise a central question: to what extent do inhibitory extensions increase the computational power of CRNs when restricted to these small-sized restricted reactions?

Under these structural restrictions, we consider two types of inhibitory CRN systems. The general \emph{Inhibitory} CRN (iCRN) \cite{calabrese2024inhibitory} model simply allows reaction execution to be inhibited by the presence of some subset of species in the configuration. The \emph{Priority-Inhibitory} CRN (P-iCRN) model establishes a priority for each reaction, in which a reaction with priority $k$ is inhibited from firing unless the first $k$ species are absent from the system. Here, the order of species is used to determine the inhibitors for each reaction. In this paper, we study the reachability problem under deletion-only and volume-preserving reactions for both of these models.

\subsection{Our Results}
For the restricted reaction types studied in this paper, we use the notation $(r,p)$ to denote reactions with $r$ reactants and $p$ products (e.g., (2,0) is a bimolecular void reaction and (1,1) is a volume-preserving unimolecular reaction).
For reference, in a traditional CRN model, reachability is in P for all of our considered void reactions \cite{AFG:2025:RRC, Fu:2025:DNA}, NL-hard for $(1,1)$-size reactions \cite{AFG:2025:RRC}, PSPACE-complete for $(2,2)$-size reactions \cite{ERW:2019:PAI}, and Ackermann-complete for general systems \cite{czerwinski2022reachability, leroux2022reachability}. 

As for organization, we first provide formal descriptions of all CRN models studied in this paper and some additional preliminaries in Section  \ref{sec:definitions}. We then study reachability under the Priority iCRN model in Section \ref{sec:p-icrns}. Here, we prove that reachability with our considered void reactions is mostly in P, but becomes NP-complete for $(2,0)+(2,1)$-size reactions. Section \ref{sec:icrns} looks at reachability for iCRNs. We show that reachability becomes NP-complete for all void reactions except for size $(1,0)$, and PSPACE-complete for $(1,1)$ and $(2,2)$-size reactions. We additionally provide FPT algorithms for solving reachability with void reactions of sizes $(2,0), (2,1)$, and $(k,k-1)$. 
The complete scope of our reachability results is shown in Table \ref{tab:results}.
Finally, in Section \ref{sec: conclusion}, we conclude with a discussion on some open problems for resolving complexity gaps of Table \ref{tab:results} and future potential work concerning CRNs using states as an unconventional yet powerful form of inhibition.

\begin{table}[t]
    \vspace{-.2cm}
    \centering
    \begin{tabular}{|@{}c@{}||c|@{}c@{}||c|c||c|c|c|}
    \hline
    \multicolumn{8}{|c|}{\textbf{Reachability under Inhibition}} \\
    \hline \hline
    \textbf{Rule Size} & \hyperref[subsec:crn]{\textbf{CRN}} & \textbf{Ref.} &\textbf{\hyperref[subsec:pvas]{\textbf{P-iCRN}}} & \textbf{Ref.} & \multicolumn{2}{c|}{\hyperref[subsec:inhib]{\textbf{iCRN}}} & \textbf{Ref.} \\
    \hline
    \multicolumn{8}{|c|}{\textbf{Deletion-Only Systems}} \\
    \hline
    $(1,0)$ & P& - &P & Obs. \ref{obs:(1,0) p-iCRN} & \multicolumn{2}{c|}{P} & Obs. \ref{obs:(1,0) iCRN} \\
    \hline
    $(2,0)$ & P& \cite{AFG:2025:RRC} & P & Thm. \ref{thm:(2,0)-Priority}  & NPC  & \cellcolor{blue!15} \hyperref[open: fpt-improvements]{FPT}& Thm. \ref{thm:(2, 0) icrn}, \ref{thm:(2,0)-fpt-icrn}  \\
    \hline
    $(2,1)$ & P& \cite{Fu:2025:DNA} &  P & Thm. \ref{thm:(k,k-1) PiCRN}  & NPC & \cellcolor{blue!15} \hyperref[open: fpt-improvements]{FPT} & Thm. \ref{thm:(2, 1) icrn}, \ref{thm:(k,k-1)-fpt-icrn} \\
    \hline
    $(k, k-1)$ & P & \cite{Fu:2025:DNA} & P & Thm. \ref{thm:(k,k-1) PiCRN} & NPC & \cellcolor{blue!15} \hyperref[open: fpt-improvements]{FPT} & Thm. \ref{thm:(k, k-1) icrns}, \ref{thm:(k,k-1)-fpt-icrn} \\
    \hline
    $(2,0), (2,1)$ &P & \cite{Fu:2025:DNA} & \cellcolor{blue!15} \hyperref[open: w1-hardness]{NPC} & Thm. \ref{thm:(2, 0) + (2, 1) picrns} & \multicolumn{2}{c|}{\cellcolor{blue!15}\hyperref[open: w1-hardness]{NPC}} &  Thm. \ref{thm:(2, 0) + (2, 1) icrns}\\
    \hline
    \multicolumn{8}{|c|}{\textbf{General Systems}} \\
    \hline
    $(1,1)$ & P, NL-hard & \cite{AFG:2025:RRC} & \multicolumn{2}{c|}{\cellcolor{blue!15} \hyperref[open: (1,1)]{Open}} & \multicolumn{2}{c|}{PSPACE-c} & Thm. \ref{thm:iCRN (1,1) Pspace-hardness} \\
    \hline
    $(2,2)$ & PSPACE-c& \cite{ERW:2019:PAI} & PSPACE-c & Obs. \ref{obs:(2,2)-Picrn} & \multicolumn{2}{c|}{PSPACE-c} & Obs. \ref{obs:(2,2)-icrn} \\    
    \hline
    General & Ack-c& \cite{czerwinski2022reachability,leroux2022reachability} & Decidable & \cite{guttenberg2024flattability,REINHARDT2008PVAS} & \multicolumn{2}{c|}{Undecidable}  & \cite{calabrese2024inhibitory,REINHARDT2008PVAS}  \\
    \hline
    \end{tabular}
    \caption{Reachability with Inhibition for restricted rule types. Void systems that are NP-complete in regular CRNs, $(k, k-i)$ where $k \geq 3$ and $i \geq 2$, are not listed since they are still NP-hard with inhibitions. Cells highlighted in \colorbox{blue!15}{blue} link to interesting open directions (in Section~\ref{sec: conclusion}) that extend the results in the table.}
    \label{tab:results}
    \vspace{-.4cm}
\end{table}

\section{Preliminaries}\label{sec:definitions}

We define the three chemical reaction network models considered in this paper: the basic CRN model (Section~\ref{subsec:crn}), the Inhibitory CRN model (Section~\ref{subsec:inhib}), and the Priority iCRN model (Section~\ref{subsec:pvas}).

\begin{figure}[t]
    \centering
    \begin{subfigure}{0.30\textwidth}
        \centering
        \includegraphics[width=1\textwidth]{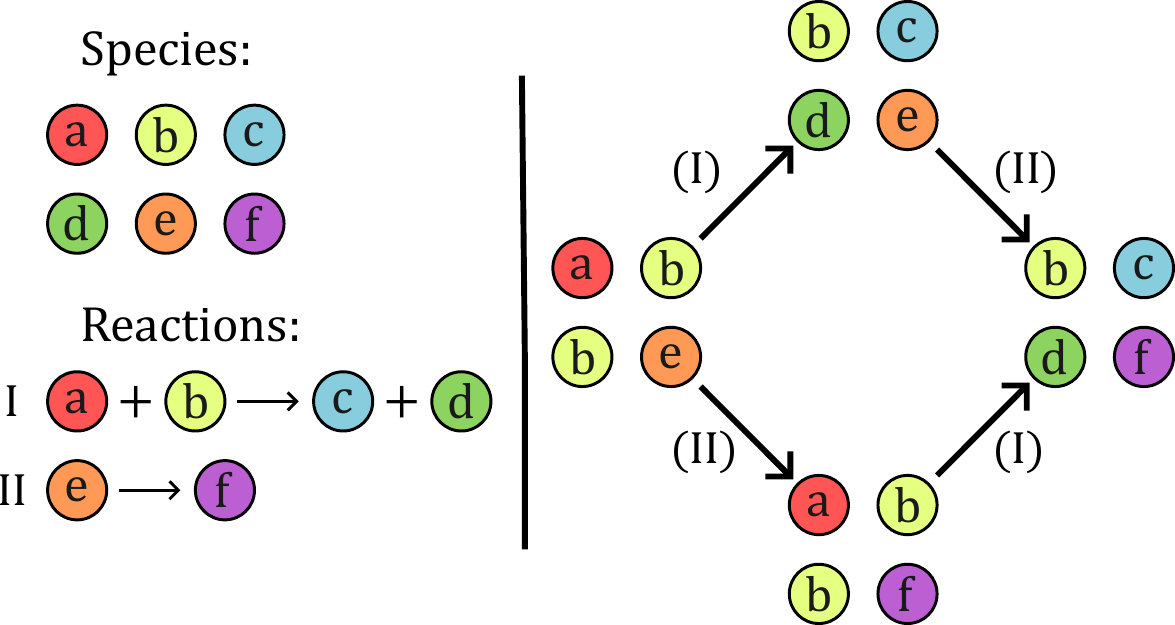}
        \subcaption{Basic CRN Example.}
        \label{subfig:basic_example}
    \end{subfigure}
    \hspace{0.2cm}
    \begin{subfigure}{0.30\textwidth}
        \centering
        \includegraphics[width=1\textwidth]{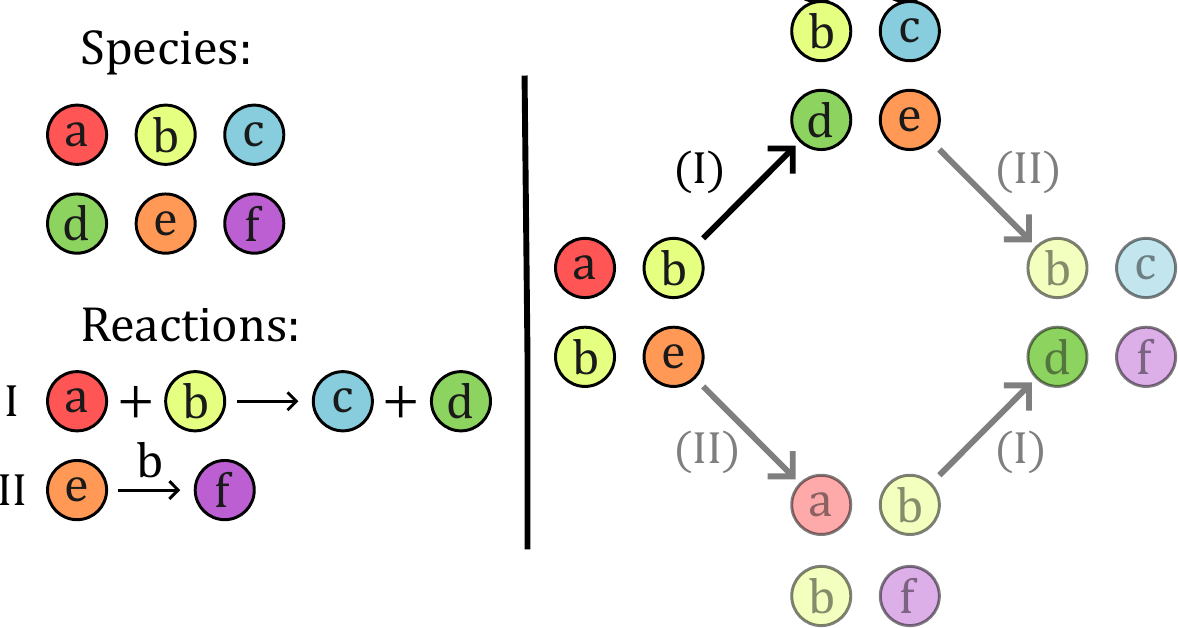}
        \subcaption{Inhibitory CRN Example.}
        \label{subfig:inhibitory_example}
    \end{subfigure}
    \hspace{0.2cm}
    
    \begin{subfigure}{0.40\textwidth}
        \centering
        \includegraphics[width=1\textwidth]{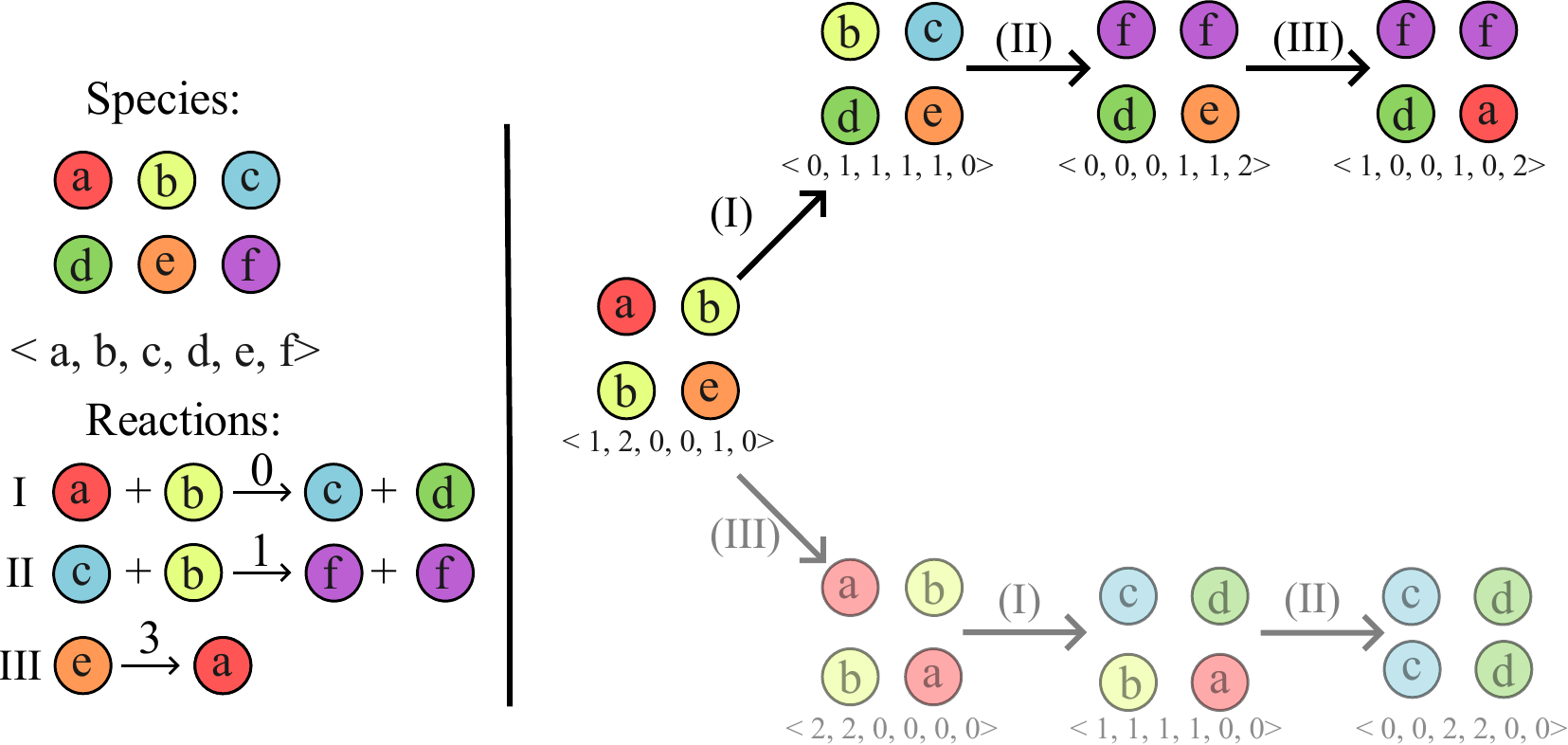}
        \subcaption{Priority iCRN Example.}
        \label{subfig:priority_example}
    \end{subfigure}
    \caption{Example systems of the CRN models studied in this paper. Each arrow represents an application of a reaction. Transparent arrows represent invalid reaction applications. (a) Basic CRN. Observe how, despite two different reaction sequences being applicable in the initial configuration, both lead to the same terminal configuration. (b) Inhibitory CRN. Because Reaction $2$ is inhibited by the species $b$, it cannot be applied in this system. (c) Priority iCRN. Reaction 3 is inhibited by species $a$ and $b$, since both species are in the system, it's not an applicable reaction. }
    \label{fig:crn_models}
\end{figure}

\subsection{Chemical Reaction Networks}\label{subsec:crn}

A \emph{chemical reaction network} (CRN) $\crn{C} = (\species, \reactions)$
is defined by a finite set of species $\species$
, and a finite set of reactions $\reactions$ where each reaction is a pair $( \reactants,\products) \in \N^\species \times \N^\species$, sometimes written $\config{R} \rxn \config{P}$, that denotes the \textit{reactant} species consumed by the reaction and the \textit{product} species generated by the reaction.
For instance, given $\species = \{a,b,c\}$, the reaction $( (2,0,0), (0,1,1) )$ represents $2a \rxn b + c$; 2 $a$ species are removed, and 1 new $b$ and $c$ species are created. See Figure \ref{subfig:basic_example} for an example.

A \textit{configuration} $\config{C}\in \N^{\size{\species}}$ of a CRN assigns integer counts to every species $\lambda \in \species$, and we use notation $\config{C} \, [\lambda]$ to denote that count.
For a species $\lambda \in \species$, we denote the configuration consisting of a single copy of $\lambda$ and no other species as $\single{\lambda}$.  
It is often useful to reference the set of species whose counts are not zero in a given configuration.
In such cases, the notation $\Set{\config{C}}$ is used.
Formally, $\Set{\config{C}} = \{\lambda \in \species \mid \config{C}[\lambda] > 0\}$, and when convenient and clear from the context, we further use $\{\config{C}\}$ to denote the configuration (vector) representation in which each element has a single copy.  Finally, let $|\config{C}| = \sum_{\lambda\in \species} C[\lambda]$ denote the total number of copies of all species in a configuration, sometimes referred to as the \emph{volume} of $\config{C}$.

A reaction $(\reactants,\products)$ is said to be \textit{applicable} in configuration $\config{C}$ if $\reactants \leq \config{C}$; in other words, a reaction is applicable if $\config{C}$ has at least as many copies of each species as $\reactants$. If the reaction $(\reactants,\products)$ is applicable, it results in configuration $\config{C'} = \config{C} - \reactants + \products$ if it occurs, and we write $\config{C} \rightarrow_{crn}^{(\species,\reactions)} \config{C'}$, or simply $\config{C} \stepsto \config{C'}$ when the model and CRN are clear from context.

A \emph{void} reaction is any rule that does not create any new copies of any species type, and can only delete or preserve existing species copies. Thus, a void reaction either has no products or has products that are a subset of its reactants (in which case these products are termed \emph{catalysts}). 

\begin{definition}[Discrete Chemical Reaction Network]  A discrete chemical reaction network (CRN) is an ordered pair $(\species, \reactions)$ where $\species$ is an ordered alphabet of species, and $\reactions$ is a set of rules over $\species$.
\vspace{-.1cm}
\end{definition}
\begin{definition}[Basic CRN Dynamics]  For a CRN $(\species, \reactions)$ and configurations $\config{A}$ and $\config{B}$, we say that $\config{A} \rightarrow_{crn}^{(\species,\reactions)} \config{B}$ if there exists a rule $(\reactants,\products)\in \reactions$ such that $\reactants \leq \config{A}$, and $\config{A} - \reactants + \products = \config{B}$.
\end{definition}

If there exists a finite sequence of configurations such that $\config{C} \stepsto \config{C}_1 \stepsto \dots \stepsto \config{C}_n \stepsto \config{D}$, then we say that $\config{D}$ is \textit{reachable} from $\config{C}$ and we write $\config{C} \reaches \config{D}$.
A configuration is said to be \textit{terminal} if no reactions are applicable.
We also define an \emph{initial configuration} for a CRN as its starting configuration. A \emph{CRN System} $T$ is then defined as a pair of a CRN model and its initial configuration.

The following sections define extensions of the basic CRN model by way of defining modified dynamics.
For each model, assume the concepts of reachability and terminality are derived from the dynamics in the same manner as for the basic CRN model.

\subsubsection{Inhibitory CRNs}\label{subsec:inhib}

A reaction $\reaction$ is said to be \textit{inhibited} by a species $\lambda$ when the reaction $\reaction$ may only be applied if $\lambda$ is absent in the system. We define an inhibitor mapping $\mathcal{I} : \reactions \rightarrow \mathbb{P}(\species)$ that maps a reaction to a subset of species that inhibit the reaction. An \emph{Inhibitory CRN} $\crn{C_{IC}} = ((\species, \reactions), \mathcal{I})$ as defined by \cite{calabrese2024inhibitory} is then a basic CRN along with the mapping $\mathcal{I}$. See Figure \ref{subfig:inhibitory_example} for an example.

\begin{definition}[Inhibitory Dynamics]\label{def:inhibitory-dynamics}
    For a Inhibitory CRN $((\species, \reactions), \mathcal{I})$ and configurations $\config{A}$ and $\config{B}$, we say that $\config{A} \rightarrow_{\crn{C_{IC}}}^{(\species, \reactions)} \config{B}$ if there exists a rule $\reaction = (\reactants, \products) \in \reactions$ such that $\reactants \leq \config{A}$, $\config{A} -\reactants+\products=\config{B}$, and $A[\lambda] = 0, \forall \lambda \in \mathcal{I}(\reaction)$.
\end{definition}

\subsection{Vector Addition Systems}
A $d$-dimensional \emph{Vector Addition System} (VAS)  $\vas{V} = (\config{C_0}, \transitions)$ is defined by \cite{karp1969parallel} as an initial configuration vector $\config{C_0} \in \Z_{\geq 0}^d$, and a finite set of transitions $\transitions \subseteq \Z^d$.

\begin{definition}[VAS Dynamics]
    For a VAS $(\config{C_s}, \transitions)$ and configurations $\config{A}$ and $\config{B}$, we say $\config{A} \rightarrow_{\vas{V}} \config{B}$ if there exists a transition $\config{y} \in \transitions$ such that $\config{A} + \config{y} = \config{B}$.
\end{definition}

The concept of priority on transitions in a Vector Addition System is defined in \cite{guttenberg2024flattability}. Each transition in a $d$-dimensional \emph{Priority Vector Addition System} (PVAS) is mapped to an integer in $[0,d]$. For a system in configuration $\config{A}$, a transition $\config{y}$ with priority $p \leq d$ is applicable iff $\forall i\leq p, \config{A}[i] = 0$ and $\config{A} + \config{y} \in \Z_{\geq 0}^d$. 

The correspondence between Vector Addition Systems and Chemical Reaction Networks is direct, especially in the absence of catalysts. A transition vector $\config{y}$ in a VAS can be written as a reaction $(\reactants, \products)$ in CRN where
\[
    (\reactants[\lambda_i], \products[\lambda_i]) = \begin{cases}
        (\big|\config{y}[i]\big|, 0) & \text{if}\; \config{y}[i] < 0\\
        (0, \config{y}[i]) & \text{if}\; \config{y}[i] > 0\\
        (0,0) & \text{otherwise}
    \end{cases}
\]

Since PVAS are essentially VAS with prioritized inhibition on transitions, it is natural to compare this model to iCRNs, which have unprioritized inhibition on reactions.
While Vector Addition Systems are not capable of \emph{catalytic} transitions, it is known that a catalytic reaction in a CRN can be replaced by two non-catalytic reactions that use a unique intermediate species.
Therefore, Vector Addition Systems are equivalent to Chemical Reaction Networks in general. However, this equivalence may break down under restricted settings such as deletion-only or small-size rules. This observation motivates a reformulation of PVAS as Priority iCRN in order to compare these results to inhibitory CRNs without priority.

\subsubsection{Priority iCRN} \label{subsec:pvas}
We define a \emph{priority} mapping $\mathcal{P}:\reactions \rightarrow [\,0,\size{\species}\,]$ that maps a reaction to an integer in $[\,0,\size{\species}\,]$. A \emph{Priority iCRN} $\crn{C_{PI}} = ((\species, \reactions), \mathcal{P})$ is defined as a basic CRN along with mapping $\mathcal{P}$. Because the priority relies on the ordering of species, we make our set of species $\species$ in a Priority iCRN an ordered set. See Figure \ref{subfig:priority_example} for an example.

\begin{definition}[Priority Inhibitory Dynamics]
    For a Priority iCRN $((\species, \reactions), \mathcal{P})$ and configurations $\config{A}$ and $\config{B}$, we say that $\config{A} \rightarrow_{\crn{C_{PI}}}^{(\species, \reactions)} \config{B}$ if there exists a rule $\reaction = (\reactants, \products) \in \reactions$ such that $\reactants \leq \config{A}$, $\config{A} -\reactants+\products=\config{B}$, and $A[\lambda_i] = 0, \forall i \in [1,\mathcal{P}(\reaction)]$.
\end{definition}

\subsection{Reachability}
\begin{definition}[Reachability Problem]
Given any CRN System with a species set $\species$ and a reaction set $\reactions$, an initial configuration $\config{C_s}$ and a destination (target) configuration $\config{C_t}$: is $\config{C_t}$ reachable from $\config{C_s}$ following the dynamics of the model. 
\end{definition}

We also consider a special variant of reachability referred to as the \emph{empty configuration reachability} problem. Here, the destination configuration is always the empty configuration $\emptyConfig$, where all species have a count of zero. Thus, the problem can be re-interpreted as asking if there exists a sequence of reactions that can delete all present species in $\config{C_s}$.
\section{Priority Inhibitory CRNs}\label{sec:p-icrns}
We start our analysis of reachability in inhibitory systems with Priority iCRNs. We first show that the problem does not leave P with size $(2,0)$ or $(2,1)$ reactions. However, when considering systems with both $(2,0)$ and $(2,1)$ size rules, we show that the problem becomes NP-complete, even with a maximum priority of 1. 

\subsection{Priority iCRNs with Only Void Rules}

We state the following observation for completeness.

\begin{observation} \label{obs:(1,0) p-iCRN}
    Given a Priority $iCRN~\crn{C_{PI}} = ((\species, \reactions), \mathcal{P})$ with only void rules of size $(1,0)$, the reachability problem is solvable in $\mathcal{O}(\size{\species} + \size{\reactions})$ time.
\end{observation}
\begin{proof}
    Given a Priority iCRN, start configuration $\config{C_s}$ and target configuration $\config{C_t}$:
    \begin{enumerate}
        \item For any species $\lambda_i$, if $\config{C_t}[\lambda_i] \not= 0$ then remove all reactions $\reaction_j$ such that $\mathcal{P}(\reaction_j) \geq i$. Since we only consider void rules, such reactions will never be applicable in any configuration along any path from $\config{C_s}$ through $\config{C_t}$.
        \item For $i = 1$ to $\size{\species}$, if there exists a reaction of form $\lambda_i \rightarrow \phi$: 
            \begin{enumerate}
                \item then apply the reaction $\config{C_s}[\lambda_i] - \config{C_t}[\lambda_i]$ times.
                \item else, skip to step $(3)$.
            \end{enumerate}
        \item If the final configuration is reached, return ``yes'', otherwise return ``no''.
    \end{enumerate}
\end{proof}

\textbf{Void Rules of Size $(2,0)$.}
\noindent
We now show that the reachability problem in Priority iCRNs with $(2,0)$ reactions is also solvable in polynomial time. To do so, we start by showing that the general reachability problem for a given Priority iCRN with $(2,0)$ void rules can be reduced to the empty reachability problem in the given CRN in Lemmas \ref{lem:piCRN-empty1}, \ref{lem:piCRN-empty2}, and \ref{lem:piCRN-empty}. We then reduce the empty configuration reachability problem in the Priority iCRN to the known maximum $b$-matching problem, which is solvable in polynomial time \cite{b-matching_runtime}, in Lemma \ref{lem:(2,0)-Priority}.

Given a $(2,0)$ void Priority iCRN system $((\species, \reactions), \mathcal{P})$ such that $k = \max_{\reaction\in\reactions} \mathcal{P}(\reaction)$. We can partition the set of species $\species$ into set of inhibiting species $\{\lambda_1,\ldots,\lambda_k\}$ and set of non-inhibiting species $\{\lambda_{k+1},\ldots,\lambda_{\size{\species}}\}$. 

\begin{lemma}\label{lem:piCRN-empty1}
    Given two configurations $\config{C_s}$ and $\config{C_t}$ in a $(2,0)$ void Priority iCRN $\crn{C_{PI}}$, $\config{C_s} \reaches \config{C_t} \iff \config{C_s}-\config{C_t}\reaches\emptyConfig$, if $\forall i\in[1,k],$ $\config{C_t}[\lambda_i] = 0$.
\end{lemma}
\begin{proof}
    Given that $\forall i\leq k, \config{C_t}[\lambda_i]=0$, $\config{C_s}[\lambda_i]-\config{C_t}[\lambda_i] = \config{C_s}[\lambda_i]$. Therefore, all inhibited reactions will be applicable at some point.

    Since void rules only delete species, and do not add new counts for species, for all non-inhibiting species $\lambda_j$ with $j>k$, if $\config{C_t}[\lambda_j]$ is reachable from $\config{C_s}[\lambda_j]$, then $\emptyConfig[\lambda_j]$ is reachable from $\config{C_s}[\lambda_j]-\config{C_t}[\lambda_j]$.
\end{proof}

\begin{lemma}\label{lem:piCRN-empty2}
    Given two configurations $\config{C_s}$ and $\config{C_t}$ in a $(2,0)$ void Priority iCRN $\crn{C_{PI}}$, $\config{C_s} \reaches \config{C_t} \iff \config{C_s}-\config{C_t} \reaches \emptyConfig$, if $\exists i \leq k, \config{C_t}[\lambda_i]\not= 0$.
\end{lemma}
\begin{proof}
    We modify the given Priority iCRN $\crn{C_{PI}}$ system as follows. We remove all reactions $\reaction \in \reactions$ such that $\mathcal{P}(\reaction) \geq i$, $\forall i \leq k$ if $\config{C_t}[\lambda_i] \not= 0$. This is because such reactions will never be applicable as they will always be inhibited. 
    Due to this modification, the species $\lambda_i,\ldots,\lambda_k$ no longer permanently act as inhibitors. The new reaction set $\reactions'$ contains the remaining inhibited reactions with priority greater than $0$. We now compute $k'$ over $\reactions'$ as $k' = \max_{\reaction\in\reactions'} \mathcal{P}(\reaction)$. We partition the set $\species$ again into set of inhibiting species $\{\lambda_1,\ldots,\lambda_{k'}\}$ and set of non-inhibiting species $\{\lambda_{k'+1},\ldots,\lambda_{\size{\species}}\}$. From Lemma \ref{lem:piCRN-empty1}, $\config{C_s} \reaches \config{C_t} \iff \config{C_s}-\config{C_t}\reaches\emptyConfig$ in a Priority iCRN $\crn{C_{PI}} = ((\species,\reactions'),\mathcal{P})$.
\end{proof}

\begin{lemma}\label{lem:piCRN-empty}
    Given two configurations $\config{C_s}$ and $\config{C_t}$ in a $(2,0)$ void Priority iCRN $\crn{C_{PI}}$, $\config{C_s}\reaches\config{C_t}\iff\config{C_s}-\config{C_t}\reaches\emptyConfig$.
\end{lemma}
\begin{proof}
    Follows from Lemmas \ref{lem:piCRN-empty1} and \ref{lem:piCRN-empty2}.
\end{proof}

\begin{definition}[$b$-matching]
Given a graph $G = (V, E)$ and some edge capacity function $u:E \rightarrow \mathbb{N}\cup\{\infty\}$, a $b$-value function $b:V \rightarrow \mathbb{N}$, and a $\delta(v)$ function that returns the set of incident edges of vertex $v$, find a maximum assignment $f:E \rightarrow \mathbb{N}$ s.t. $f(e) \leq u(e)$ for all $e \in E$ and $\sum_{e \in \delta(v)}f(e) \leq b(v)$ for all $v \in V$. We call it a perfect $b$-matching if $\sum_{e \in \delta(v)}f(e) = b(v)$ for all $v \in V$.
\end{definition}
The runtime of the maximum $b$-matching problem is strictly polynomial and runs in $O(|V|^2 \log (|V|)(|E| + |V| \log (|V|)))$ \cite{b-matching_runtime}. We use this runtime to solve the reachability problem for Priority iCRNs with only $(2,0)$-size reactions in polynomial time.

\begin{lemma}\label{lem:(2,0)-Priority}
Given a Priority $iCRN~\crn{C_{PI}} = ((\species, \reactions), \mathcal{P})$ with only void rules of size $(2,0)$, the empty reachability problem is decidable in $O(\size{\species}^2\log(\size{\species})(\size{\reactions} + \size{\species}\log(\size{\species})))$ time.

\end{lemma}
\begin{proof}
    To solve the empty-configuration reachability problem in $(2,0)$ void Priority iCRN, we reduce it to an instance of the maximum $b$-matching problem.

    Given a Priority iCRN $\crn{C_{PI}} = ((\species, \reactions), \mathcal{P})$ with starting configuration $\config{C_s}$, we construct a $b$-matching instance by building a graph $G=(V,E)$ as follows. Here we assume that the given Priority iCRN only includes reactions $\reaction: \lambda_i + \lambda_j \rightarrow \phi$ where $i, j > \mathcal{P}(\reaction)$; any self-inhibiting reactions that exist ($i$ or $j$ is less than the priority) are ignored. We turn every species $\lambda_i\in\species$ into a vertex $v_i\in V$ and set the $b$-values for each of the vertices $b(v_i) = \config{C_s}[\lambda_i]$. For each $\reaction: \lambda_i + \lambda_j \rightarrow \phi$ in $\reactions$ such that $i\not=j$, add an undirected edge $e = (v_i, v_j)$ in $E$. For any reaction of the form $\lambda_i + \lambda_i \rightarrow \phi$, add vertices $v_i^1, v_i^2$ to the set $V$ and set the b-values $b(v_i^1) = b(v_i^2) = \floor{\frac{b(v_i)}{2}}$. We also add undirected edges $(v_i,v_i^1), (v_i,v_i^2), (v_i^1,v_i^2)$ in $E$.

    From the rules above, we have the graph $G$ well defined. 
    We now present the following claims regarding $G$.

    {\bf Claim 1. If $\config{C_{s}}$ can reach the empty configuration, then there exists a perfect $b$-matching in the graph $G$.} This is straightforward for the case when applying a reaction of the form $\lambda_i + \lambda_j \rightarrow \phi$ such that $i \not= j$. The reaction application results in the removal of one copy of species $\lambda_i$ and $\lambda_j$. This corresponds to an edge matching of $(v_i,v_j)$ on the graph $G$ according to the construction above. We reduce the values of $b(v_i)$ and $b(v_j)$ by one, and the pair $(v_i,v_j)$ form a perfect matching. For the case when applying a reaction of form $\lambda_i + \lambda_i \rightarrow \phi$, two copies of $\lambda_i$ are deleted from the configuration. This rule application corresponds to selecting edges $(v_i, v_i^1)$ and $(v_i, v_i^2)$. We reduce the counts of $b(v_i^1)$ and $b(v_i^2)$ by one, and that of $b(v_i)$ by two. The remaining $b$-values of vertices $v_i^1$ and $v_i^2$ are reduced by selecting edge $(v_i^1, v_i^2)$. And the pairs $(v_i, v_i^1)$, $(v_i, v_i^2)$, and $(v_i^1, v_i^2)$ form a perfect matching. 

    {\bf Claim 2. If there exists a perfect $b$-matching in the graph $G$, then $\config{C_{s}}$ can reach the empty configuration.} If there exists a perfect $b$-matching for the graph $G$, there exists an assignment $f : E \rightarrow \N$ such that $\sum_{e\in \delta(v)}f(e) = b(v)$ for all $v \in V$. From the construction of the graph, all edges $(v_i,v_j)\in E$ such that $i\not=j$ correspond to a reaction $\lambda_i + \lambda_j \rightarrow \phi$ in $\reactions$ in the Priority iCRN. And, edges of the form $(v_i,v_i^1), (v_i,v_i^2)\in E$ correspond to a reaction $\lambda_i + \lambda_i \rightarrow \phi$ in $\reactions$. We argue that, given any perfect $b$-matching, we can apply a corresponding set of reactions in a rearranged order based on their priority to reach the empty configuration. 
    
    Let $M$ contain the set of edges to be a perfect $b$-matching for $G$. We apply reactions in Priority iCRN corresponding to the edges in $M$ as follows.
    \begin{enumerate}
        \item Apply all reactions $\reaction = \lambda_i + \lambda_j \rightarrow \phi$ corresponding to edges $e_k=(v_i, v_j)$, $f(e_k)$ times, if $e_k\in M$, $i\not = j$ and $\mathcal{P}(\reaction)=0$. All such reactions are uninhibited and can be applied.
        \item Apply all reactions $\reaction = \lambda_i + \lambda_i \rightarrow \phi$ corresponding to edges $e_k=(v_i, v_i^1)$ and $e_{\ell}=(v_i, v_i^2)$, $f(e_k)$ times, if $e_k, e_{\ell} \in M$ and $\mathcal{P}(\reaction)=0$. The vertices $v_i^1$ and $v_i^2$ are only connected to each other and to vertex $v_i$. To achieve perfect b-matching for vertices $v_i^1$ and $v_i^2$, assigned values $f(e_k)$ and $f(e_{\ell})$ must be exactly equal, $f(e_k) = f(e_{\ell)} = b(v_i^1)-f(e_m)$ where $e_m=(v_i^1, v_i^2)$. Applying such a reaction will delete $2*f(e_k)$ copies of species $\lambda_i$.
        \item For all $p$ in $1$ through $\max_{\reaction\in\reactions}\mathcal{P}(\reaction)$: 
        \begin{enumerate}
            \item apply all reactions $\reaction = \lambda_i + \lambda_j \rightarrow \phi$ corresponding to edges $e_k=(v_i, v_j)$, $f(e_k)$ times, if $e_k\in M$, $i\not = j$ and $\mathcal{P}(\reaction)=p$.
            \item apply all reactions $\reaction = \lambda_i + \lambda_i \rightarrow \phi$ corresponding to edges $e_k=(v_i, v_i^1)$ and $e_{\ell}=(v_i, v_i^2)$, $f(e_k)$ times, if $e_k, e_{\ell}\in M$, and $\mathcal{P}(\reaction)=p$.
        \end{enumerate}
    \end{enumerate}
    
    Furthermore, the reactions applied in each iteration in Step $3$ are guaranteed to be applicable (absence of inhibitors and sufficient count of reactants) because $(1)$ all species count go to zero; $(2)$ no reaction that uses a reactant $\lambda_i$ has a priority $p > i$; and $(3)$ each reaction $e_k$ is only applied $f(e_k)$ times. If multiple reactions use a species $\lambda_i$, the total applications of all those reactions $\sum_{e\in \delta(v_i)}f(e)$ will be exactly $b(v_i)$. Therefore, rearranging applications of each of those reactions will not affect the availability of reactants in later applications.
\end{proof}

\begin{theorem}\label{thm:(2,0)-Priority}
Given a Priority $iCRN~\crn{C_{PI}} = ((\species, \reactions), \mathcal{P})$ with only void rules of size $(2,0)$, the reachability problem is decidable in  $O(\size{\species}^2\log(\size{\species})(\size{\reactions} + \size{\species}\log(\size{\species})))$ time.
\end{theorem}
\begin{proof}
    Given a Priority iCRN with starting configurations $\config{C_s}$ and target configuration $\config{C_t}$, create the configuration $\config{C_s}-\config{C_t}$. From Lemma \ref{lem:(2,0)-Priority}, empty configuration reachability is decidable from $\config{C_s}-\config{C_t}$ in polynomial time. Finally, from Lemma \ref{lem:piCRN-empty}, empty configuration reachability in $(2,0)$ Priority iCRN implies reachability.
\end{proof}

\begin{lemma} \label{lem:(k,k-1) general}
Given a $CRN~\crn{C} = (\species, \reactions)$ with only void rules of size $(k, k-1)$, the reachability problem is decidable in $\mathcal{O}{(|\species|^2|\reactions|)}$ time (Complete proof in \cite{Fu:2025:DNA}).
\end{lemma}
\begin{proof} \textit{(Sketch.)}
    The dynamic programming algorithm uses a $|\species| \times (|\species|+1)$ table $D$ of boolean entries, where each row represents a different species. Since each $(k,k-1)$ catalytic void reaction reduces the count of exactly one species, only one of the reactions need to be applied at most $\config{C_s}[\lambda_i]-\config{C_t}[\lambda_i]$ times to reduce the count of a species $\lambda_i$. Hence, only $\size{\species}$ distinct reactions need to be considered to reach the target configurations, if possible. The total number of rule applications, however, will be polynomial in the volume of the system.
    
    A bottom-up approach is used to determine the ordering of the species in which their count needs to be reduced. This is essential since a species $\lambda_i$ shouldn't go to zero before it is used as a catalyst in other reactions that need to be applicable.
    The algorithm starts by filling first column of $D$ with $1$ for species who have reached their target count. It then progresses column-by-column, applying reactions based on the ordering on species. If the last column has all $1$s, then the configuration is reached.
\end{proof}

\begin{theorem} \label{thm:(k,k-1) PiCRN}
Given a Priority $iCRN~\crn{C_{PI}} = ((\species, \reactions), \mathcal{P})$ with only void rules of size $(k, k-1)$, the reachability problem is decidable in $\mathcal{O}{(|\species|^2|\reactions|)}$ time.
\end{theorem}

\begin{proof}
    This follows as a modification of the dynamic programming algorithm for general CRNs with only void rules of size $(k, k-1)$ from Lemma \ref{lem:(k,k-1) general}. Let $\config{C_f}$ denote the target configuration, with $\lambda_1, \ldots, \lambda_{|\species|}$ denoting the order assigned to the species. Remove all rules of the form $\lambda_j + (.) \rightarrow^{k} (.)$, where $j < k$. The rest of the system remains the same.
    
    Start with a counter $l=|\Lambda|$. Construct an $|\species| \times (|\species|+1)$ table $D(s,j)$ of boolean entries, where each row represents a different species. Reduce the count of each species to $max(k, s^f)$, where $s^f$ represents the final count of species $s$, if there exists a rule that can do so. Starting from the first column $j=0$, place a $1$ if the respective species is already in its final count. Then, for each entry $D(s,j)$, place a $1$ if $D(s,j-1)$ is a $1$, or if there exists a reaction $\reaction \in \reactions$ for $s \in \species$ that reduces $s$ to its final count such that 1) $\config{C_f}[s] > 0$ or 2) $s = \lambda_l$, where all the reactants of $\reaction$ have either reached their final counts or will not prevent the reaction from occurring once they do. Any time species $\lambda_l$ has found a reaction to reach a final count $0$, set $l = l - 1$. If column $|\species|+1$ contains all 1's, then reachability is possible. Otherwise, it is not. 

    The algorithm follows the same bottom-up approach as for general CRNs. The main difference is that now there exists a predetermined ordering that species must go to zero for other rules to be applicable. Thus, any species $\lambda_j$ such that $j < l$ is not considered until species $\lambda_l$ has found a suitable reaction and reaches a count of zero. The proof then follows from general CRNs.
\end{proof}

\begin{lemma}[Rearrangement Lemma for Void (Priority) iCRNs]\label{lem:rearrange}
    For any sequence of applicable void rules $A$ in a given iCRN, there exists a sequence $B$ that is a permutation of $A$ such that all applications of a given rule type occur contiguously.
\end{lemma}
\begin{proof}
    Consider a sequence of applicable void rules $A$ for a given (Priority) iCRN that is not contiguous. We construct a contiguous sequence as described in the Rearrangement Lemma in \cite{Fu:2025:DNA}. For the given sequence $A$, suppose rule $\reaction$ occurs at positions $i$ and $j$ such that $i < j-1$, and there is at least one non-$\reaction$ rule in between $i$ and $j$. Construct a sequence $A'$ by shifting the rule $\reaction$ at position $i$ up to position $j-1$, and shifting all rules in between down one position. 
    
    This new sequence must be applicable as $(1)$ the only rule that moved to a higher index in the sequence is of type $\reaction$, and we know that $\reaction$ is applicable at position $j-1$ since it is applicable at position $j$, and $(2)$ because, in a system with just void rules, species are only deleted and not produced. Therefore, all rules that occur between positions $i$ and $j$ must not be inhibited by the reactants of $\reaction$ (since the reactants of $\reaction$ are present until position $j$); therefore, such rules can be applied before all occurrences of $\reaction$.
    
    Since this swapping preserves the applicability of the sequence while reducing the number of non-contiguous blocks of one rule type in the sequence, we can repeat this process of swapping rule positions until the sequence is contiguous.
\end{proof}

\begin{lemma}\label{thm:NP-voidiCRN}
    The reachability problem for (Priority) iCRNs with only void reactions is in NP.
\end{lemma}
\begin{proof}
    Given a sequence of applicable void rules in a given iCRN, we encode this sequence with a sequence of rule types accompanied by a count on the number of applications of each rule type, which must exist by Lemma \ref{lem:rearrange}. Although the original sequence is potentially exponential in length, the result of the contiguous sequence can be computed in polynomial time. Therefore, we utilize a contiguous sequence of applicable rules as a certificate for the reachability problem.
\end{proof}

\begin{restatable}{theorem}{picrnVC}\label{thm:(2, 0) + (2, 1) picrns}
    The reachability problem in Priority iCRN with only $(2,0)$ and $(2,1)$ void rules with maximum priority $1$ is NP-complete.
\end{restatable}

\begin{proof}
    Due to space constraints, we provide a proof sketch here, while the complete proof can be found in Appendix \ref{sec:np-hard-picrn}.

    From Lemma \ref{thm:NP-voidiCRN}, we have that the reachability problem for Priority iCRNs with purely void rules is in NP. To show hardness, we reduce from Vertex Cover. Given an instance of VC $\langle G,k \rangle$ where $k>0$, we construct a Priority iCRN $\crn{C_{PI}}$. We create one copy of $v$ for all $v\in V$ and one copy of $e$ for all $e\in E$. We also add two new species to the $\crn{C_{PI}}$ system. In the initial configuration, we add $|V|-k$ copies of a species $R$, used to eliminate $|V|-k$ vertices, and $k$ copies of a species $X$, used to remove the remaining $k$ vertices. Denote this configuration the starting configuration $\config{C_s}$. We also order our species set as: $\species = \langle R, v_1,\ldots,v_{\size{V}}, e_1,\ldots,e_{\size{E}}, X\rangle$.

    For all $v\in V$, we create the \emph{assignment reaction} $v+R\rightarrow\emptyset$ and the \emph{clean-up reaction} $v+X\xrightarrow{1}\emptyset$. For all edges $e\in E$, for all vertices $v\in V$ where $v$ is incident to $e$, we create the \emph{covering reaction} $v+e\xrightarrow{1}v$. Each of the clean-up and covering reactions have priority $1$.

    The target configuration $\config{C_t}$ is the empty configuration $\emptyConfig$. A vertex cover of size $k$ can only exist in $G$ iff all the species copies can be completely deleted.
\end{proof}

\begin{lemma}\label{lem:pspace}
    The reachability problem in (Priority) iCRNs with only reactions of size $(2,2)$ is in PSPACE.
\end{lemma}
\begin{proof}
    Given that the system is volume-preserving, the reachability problem is solvable by a non-deterministic Turing machine using only polynomial space, where the tape of the machine stores a configuration (counts of species). Therefore, the problem is in NPSPACE, and hence in PSPACE \cite{SAVITCH1970177}.
\end{proof}

\begin{observation}\label{obs:(2,2)-Picrn}
    The reachability problem in Priority iCRNs with only reactions of size $(2,2)$ is PSPACE-complete.
\end{observation}
\begin{proof}
    The hardness follows from the reachability problem in regular CRNs with only $(2,2)$-size rules being PSPACE-complete \cite{ERW:2019:PAI}, and membership follows from Lemma \ref{lem:pspace}.
\end{proof}
\section{Inhibitory CRNs}\label{sec:icrns}
We now shift our focus to the reachability problem in the iCRN model. Specifically, we prove that reachability of iCRNs with void rules of size $(2,0)$, $(2,0)$ + $(2,1)$, and more generally $(k,k-1)$ are all NP-complete. Alongside these hardness results, we present FPT algorithms for each of these reachability problems. To conclude, we examine iCRN systems with rules of size $(1,1)$ and show that the reachability problem is PSPACE-complete. 

\subsection{iCRNs with Only Void Rules}

For completeness, we include the following observation.

\begin{observation} \label{obs:(1,0) iCRN}
Given an $iCRN~\crn{C_{IC}} = ((\species, \reactions), \mathcal{I})$ with only void rules of size $(1, 0)$, the reachability problem is decidable in $\mathcal{O}(|\species||\reactions|)$ time.
\end{observation}

\begin{proof}
    Run any applicable rule in the system, reducing the corresponding species to its final count. If the final configuration is reached, then reachability is possible. Otherwise, if no more rules can be applied, reachability is not possible.
\end{proof}

\begin{figure}
    \centering
    \begin{subfigure}[b]{0.22\textwidth}
        \centering
        \includegraphics[width=1\textwidth]{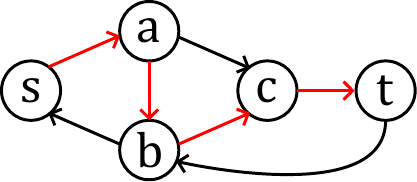}
        \caption{}
        \label{subfig:(2,0)-icrn-reduction-1}
    \end{subfigure}
    \hfill
    \begin{subfigure}[b]{0.42\textwidth}
        \centering
        \includegraphics[width=1\textwidth]{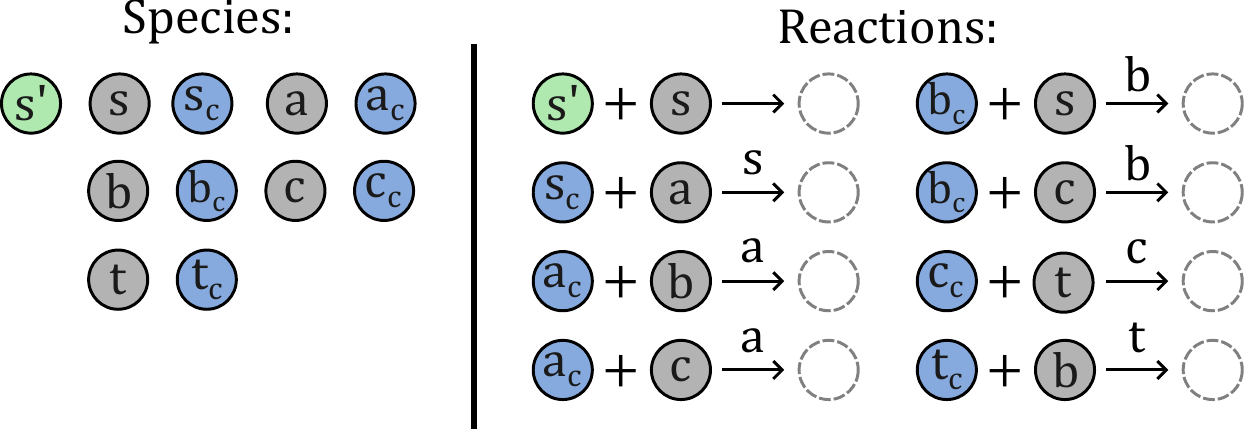}
        \caption{}
        \label{subfig:(2,0)-icrn-reduction-2}
    \end{subfigure}
    \hfill
    \begin{subfigure}[b]{0.24\textwidth}
        \centering
        \includegraphics[width=1\textwidth]{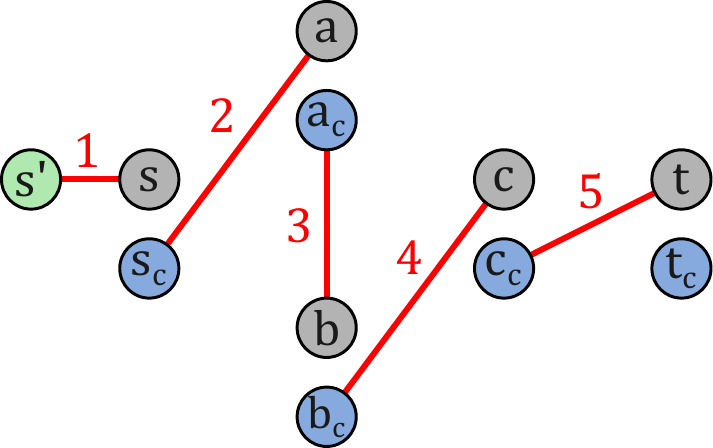}
        \caption{}
        \label{subfig:(2,0)-icrn-reduction-3}
    \end{subfigure}
    \caption{Transforming an instance of HAMPATH into an instance of iCRN reachability with only $(2,0)$-size reactions. (a) is a directed graph $G$ with a Hamiltonian path highlighted in red. (b) is the set of species and reactions of $\crn{C_{IC}}$ as constructed from $G$. (c) is the complete reachability reduction, as well as showing how the Hamiltonian path is traversed with $\crn{C_{IC}}$. The red line between two species copies represents their deletion by a $(2,0)$ void reaction; the number represents the order the reaction was applied.}
    \label{fig:(2,0)-icrn-reduction}
\end{figure}

\begin{restatable}{theorem}{icrnhampath}\label{thm:(2, 0) icrn}
    The reachability problem in iCRNs with only rules of size $(2,0)$ is NP-complete.
\end{restatable}

\begin{proof}
    Due to space constraints, we provide a proof sketch here, while the complete proof can be found in Appendix \ref{sec:np-hard-appendix}.
    
    Lemma \ref{thm:NP-voidiCRN} shows reachability in void iCRNs is in NP. To show hardness, we reduce from the Hamiltonian path problem. Given a HAMPATH instance $\langle G,s,t \rangle$, we transform it into an instance of reachability with an iCRN $\crn{C_{IC}}$. For each vertex $v \in V$, we create the species $v$ and $v_c$. We also add an extra species $s'$. For each directed edge $\{u,v\} \in E$, we create the \emph{choosing reaction} $u_c+v\xrightarrow{u}\emptyset$. We also create the \emph{starting reaction} $s'+s\rightarrow\emptyset$. Let the initial configuration $\config{A}$ be a single copy of each species and the target configuration $\config{B}$ be a single copy of $t_c$. A Hamiltonian path from $s$ to $t$ can only exist in $G$ iff the system can delete all copies except for $t_c$.
    An example reachability instance is illustrated in Figure \ref{fig:(2,0)-icrn-reduction}.
\end{proof}

We now explore the reachability problem for a $(2,0)$ void iCRN with a constant number of inhibiting species. We follow the same format as in Priority iCRN section to show that the reachbility problem is polynomial time solvable for a constant number of inhibitors. To do so we show that reachability problem in such an iCRN can be reduced to the empty reachability problem in Lemmas \ref{lem:iCRN-empty1}, \ref{lem:iCRN-empty2}, and \ref{lem:iCRN-empty}. Finally, we reduce the empty reachability problem in the given iCRN to the empty reachability problem in Priority iCRN with $(2,0)$ reactions.

Given a $(2,0)$ void iCRN $\crn{C_{IC}} = ((\species, \reactions), \mathcal{I})$, where the set of species $\Lambda$ is partitioned into the set of inhibiting species $\species_{\mathcal{I}} = \bigcup_{\reaction\in\reactions} \mathcal{I}(\reaction)$ and the set of remaining species $\species \setminus \species_{\mathcal{I}}$. We use $\emptyConfig$ to denote an empty configuration where all species have a count of zero.

\begin{lemma}\label{lem:iCRN-empty1}
    Given two configurations $\config{C_s}$ and $\config{C_t}$ in a (2,0) void iCRN $\crn{C_{IC}}$, $\config{C_s}\reaches \config{C_t}\iff\config{C_s}-\config{C_t}\reaches\emptyConfig$, if $\forall \lambda_i \in \species_{\mathcal{I}}, \config{C_t}[\lambda_i]=0$.
\end{lemma}
\begin{proof}
    Follows from Lemma \ref{lem:piCRN-empty1} because the proof does not rely on any ordering of inhibitors.
\end{proof}
\begin{lemma}\label{lem:iCRN-empty2}
    Given two configurations $\config{C_s}$ and $\config{C_t}$ in a (2,0) void iCRN $\crn{C_{IC}}$, $\config{C_s}\reaches \config{C_t}\iff\config{C_s}-\config{C_t}\reaches\emptyConfig$, if $\exists \lambda_i \in \species_{\mathcal{I}}, \config{C_t}[\lambda_i]\not=0$.
\end{lemma}
\begin{proof}
    The proof is similar to that of Lemma \ref{lem:piCRN-empty2}. Here we remove any reactions inhibited by species $\lambda_i \in \species_{\mathcal{I}}$. Now $\lambda_i$ is no longer an inhibitor. The remaining set of reactions is $\reactions'$. From Lemma \ref{lem:iCRN-empty1}, $\config{C_s}\reaches \config{C_t}\iff\config{C_s}-\config{C_t}\reaches\emptyConfig$ in a iCRN $\crn{C_{IC}} = ((\species,\reactions'),\mathcal{I})$.
\end{proof}
\begin{lemma}\label{lem:iCRN-empty}
    Given two configurations $\config{C_s}$ and $C_t$ in a (2,0) void iCRN $\crn{C_{IC}}$, $\config{C_s}\reaches \config{C_t}\iff\config{C_s}-\config{C_t}\reaches\emptyConfig$.
\end{lemma}
\begin{proof}
    Follows from Lemmas \ref{lem:iCRN-empty1} and \ref{lem:iCRN-empty2}.
\end{proof}

\begin{restatable}{lemma}{orderediCRN}\label{lem:iCRN-piCRN}
    Given any iCRN system $(\crn{C_{IC}}, \config{C_s})$ and an ordering on inhibitors, we can construct a Priority iCRN system $(\crn{C_{PI}}, \config{C_s'})$ that contains the same inhibitors in the given ordering. This construction takes $O(\size{\species}\cdot\size{\reactions})$ where $\species$ and $\reactions$ are the set of species and reactions in the given iCRN.
\end{restatable}
\begin{proof}\textit{(Sketch.)}
    We construct an ordered species set $\species'$ containing all species from the given iCRN. The species are arranged such that all the inhibitors (in order) appear first in the set. The configuration and reactions are also rearranged to match the ordering of species in $\species'$. For each reaction, the priority is set to be the maximum index of all inhibitors. Finally, we remove any reactions that are self-inhibiting, i.e., the priority of the reaction is greater or equal to the index of at least one of the reactants. The complete proof can be found in Appendix \ref{sec:np-hard-picrn}.
\end{proof}

\begin{lemma}\label{lem:(2,0)-O(1)-iCRN}
Given an $iCRN~\crn{C_{IC}} = ((\species, \reactions), \mathcal{I})$ with only void rules of size $(2, 0)$, the empty configuration reachability problem is decidable in  $O(c!\cdot f(\species,\reactions))$ time where $c$ is the number of inhibitors and $f(\species,\reactions)= \size{\species}^2\log(\size{\species})(\size{\reactions} + \size{\species}\log(\size{\species}))$.
\end{lemma}
\begin{proof}
    Given an iCRN $\crn{C_{IC}}=((\species,\reactions),\mathcal{I})$ with $(2,0)$ void rules and some start configuration $\config{C_s}$. Given that the iCRN has at most $c$ inhibitors, there exist at most $c!$ different orderings in which the inhibitors are removed.

    To solve the empty configuration reachability problem in iCRN, we reduce it to $c!$ instances of the empty configuration reachability problem in Priority iCRN. From Lemma \ref{lem:(2,0)-Priority} we know that this problem is solvable in $O(f(\species,\reactions))$. For each ordering $\langle i_1,i_2,\ldots,i_c\rangle$ where $i_j$ is some species that acts as an inhibitor, we construct a Priority iCRN as described in Lemma \ref{lem:iCRN-piCRN}. We now show that the empty configuration is reachable in the given iCRN iff the empty configuration is reachable in at least one of the $c!$ Priority iCRNs.

    {\bf Claim 1} If the empty configuration is reachable in iCRN, then it is also reachable in at least one of the constructed Priority iCRNs. This is true simply by construction of the Priority iCRNs. If the iCRN reaches the empty configuration, then all species, including inhibitor species, go to zero. Therefore, there must exist an ordering in which these inhibitors are deleted. Since we consider every possible ordering, at least one of the $c!$ Priority iCRNs orderings will reach the empty configuration.

    {\bf Claim 2} If the empty configuration is reachable in at least one of the constructed Priority iCRNs, then the iCRN reaches the empty configuration. Let the ordering corresponding to the Priority iCRN for which the empty configuration is reachable be $\langle i_1,\ldots,i_c\rangle$. Modify the reaction set $\reactions$ of the given iCRN by removing all reactions of the form $i_j + \_ \rightarrow \phi$ or $\_ + i_j \rightarrow \phi$ if they are inhibited by species $i_{k>j}$ since $i_j$ is deleted before $i_k$. For the remaining reactions, apply the reactions in the same order in which their corresponding reactions in $\crn{C_{PI}}$ are applied. If none of the $c!$ Priority iCRNs reach the empty configuration, then the given iCRN will also not reach the empty configuration.
\end{proof}

\begin{theorem}\label{thm:(2,0)-fpt-icrn}
Given an $iCRN~\crn{C_{IC}} = ((\species, \reactions), \mathcal{I})$ with only void rules of size $(2, 0)$, the reachability problem is decidable in $O(c!\cdot f(\species,\reactions))$ time where $c$ is the number of inhibitors and $f(\species,\reactions)= \size{\species}^2\log(\size{\species})(\size{\reactions} + \size{\species}\log(\size{\species}))$.
\end{theorem}
\begin{proof}
    Similar to the case of $(2,0)$ Priority iCRN, proof follows from Lemmas \ref{lem:iCRN-empty} and \ref{lem:(2,0)-O(1)-iCRN}.
\end{proof}

\begin{figure}
    \centering
    \begin{subfigure}[b]{0.12\textwidth}
        \centering
        \includegraphics[width=1\textwidth]{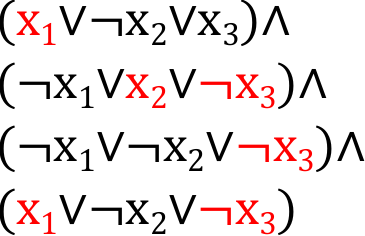}
        \caption{}
        \label{subfig:(2,1)-icrn-reduction-1}
    \end{subfigure}
    \hfill
    \begin{subfigure}[b]{0.6\textwidth}
        \centering
        \includegraphics[width=1\textwidth]{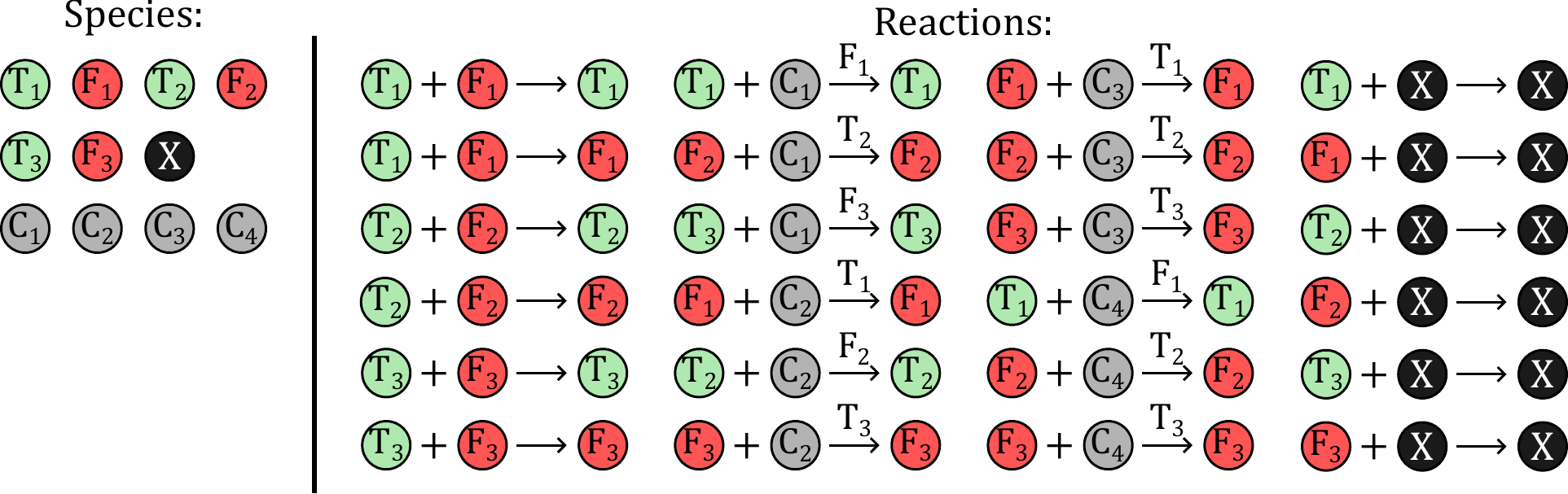}
        \caption{}
        \label{subfig:(2,1)-icrn-reduction-2}
    \end{subfigure}
    \hfill
    \begin{subfigure}[b]{0.2\textwidth}
        \centering
        \includegraphics[width=1\textwidth]{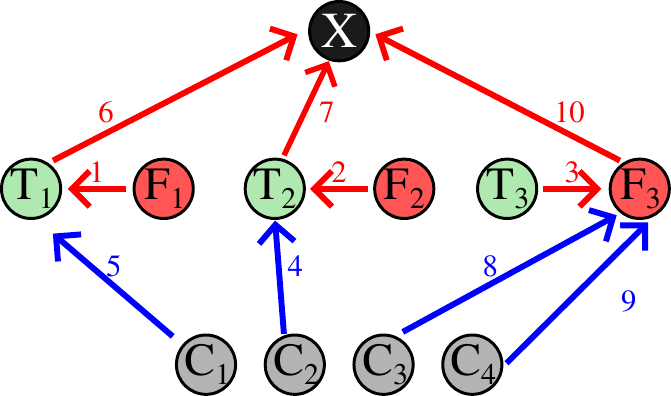}
        \caption{}
        \label{subfig:(2,1)-icrn-reduction-3}
    \end{subfigure}
    \caption{Transforming an instance of 3SAT into an instance of iCRN reachability with only $(2,1)$-size reactions. (a) is a $3$-CNF formula $\Phi$. A satisfying assignment of $\Phi$ is represented by the red letters. (b) is the set of species and reactions of $\crn{C_{IC}}$ as constructed from $\Phi$. (c) is the complete reachability reduction, as well as showing how the formula is satisfied with $\crn{C_{IC}}$. The red line between two species copies represents their deletion by a $(2,1)$ void reaction.}
    \label{fig:(2,1)-icrn-construction}
\end{figure}

\begin{restatable}{theorem}{icrnsat}\label{thm:(2, 1) icrn}\vspace{-.3cm}
    The reachability problem in (2,1) void iCRN is NP-complete.
\end{restatable}
\begin{proof}
    Due to space constraints, we provide a proof sketch here, while the complete proof can be found in Appendix \ref{sec:np-hard-appendix}.

    From Lemma \ref{thm:NP-voidiCRN}, we have that the reachability problem in void iCRNs is in NP. To show hardness, we reduce from the 3-satisfiability problem. Given an instance of 3SAT $\langle\Phi\rangle$, we create a new iCRN $\crn{C_{IC}}$. For each variable $x_i\in X_n$, we create the species $T_i$ and $F_i$. Additionally, for each clause $c_j\in C_m$, we create the species $c_j$. Finally, we create the species $X$.  For each variable $x_i$, we also create the \emph{assignment reactions} $T_i + F_i \rightarrow T_i$ and $T_i + F_i \rightarrow F_i$. For each clause $c_j$ and a variable of $c_j$ $x_k$, we create the \emph{clause reaction} $C_i + T_k \xrightarrow{F_k} T_k$ if assigning true to $x_k$ satisfies $c_j$, or $C_i + F_k \xrightarrow{T_k} F_k$ if the false assignment satisfies instead. Finally, for each variable, we create the \emph{clean-up reaction} $T_i + X \rightarrow X$ and $F_i + X \rightarrow X$. The initial configuration $\config{C_s}$ is a single copy of each species and the target configuration $\config{C_t}$ is a single copy of $X$. $\Phi$ can only be satisfied iff the system can delete all copies except for $X$.    
    An example reachability instance is illustrated in Figure \ref{fig:(2,1)-icrn-construction}.
\end{proof}

\begin{theorem} \label{thm:(k, k-1) icrns}
    The reachability problem in $(k,k-1)$ void iCRN is NP-complete.
\end{theorem}

\begin{proof}
    Follows from Theorem \ref{thm:(2, 1) icrn}. Given an iCRN $\crn{C_{IC}}$, we modify it by adding a species $d$ and modifying each reaction to include $k-2$ copies of $d$ for both reactants and products (i.e., $T_x+F_x+d\rightarrow T_x+d$ for $k=3$). We include $k-2$ copies of $d$ for $\config{C_s}$ and $\config{C_t}$. Since $d$ acts as a catalyst species for all reactions, its count will never change, and it does not affect the behavior of $\crn{C_{IC}}$. Thus, the forward and reverse directions from Lemma \ref{thm:(2, 1) icrn} follow here.
\end{proof}

\begin{theorem}\label{thm:(k,k-1)-fpt-icrn}
Given an $iCRN~\crn{C_{CI}} = ((\species, \reactions), \mathcal{I})$ with only void rules of size $(k, k-1)$, the reachability problem is decidable in  $\mathcal{O}{(c! \cdot |\species|^2|\reactions|)}$ time, where $c$ is the number of inhibitors.
\end{theorem}

\begin{proof}
    This follows by using the algorithm for Priority iCRN void systems with size $(k, k-1)$ rules from Theorem \ref{thm:(k,k-1) PiCRN} as a subroutine. Denote this algorithm $g$. Consider an arbitrary ordering of the $c$ inhibitors $\lambda_1, \ldots, \lambda_c$. Run algorithm $g$ with this ordering. If a solution is found, then reachability is possible. Otherwise, consider another ordering until all are considered. The algorithm then runs in $\mathcal{O}(c! \cdot |\species|^2|\reactions|)$ time.
\end{proof}

We show that reachability is NP-complete with $(2,0)$ and $(2,1)$ reactions, even if only one species is used to inhibit reactions.

\begin{theorem} \label{thm:(2, 0) + (2, 1) icrns}
    The reachability problem in (2,0) and (2,1) void iCRNs is NP-complete, even with a single inhibitor species.
\end{theorem}

\begin{proof}
From Lemma \ref{thm:NP-voidiCRN}, we have that the reachability for iCRNs with purely void rules is in NP. We now show hardness with a similar reduction as in Theorem \ref{thm:(2, 0) + (2, 1) picrns}. We construct our iCRN $\crn{C_{IC}}$ in a similar manner, where our species set is now unordered. Here, the species set $\species = \{R,v_1,\ldots,v_{\size{V}},e_1,\ldots,e_{\size{E}}, X\}$. Our reaction set contains \emph{assignment reactions} $v + R \rightarrow \phi$, \emph{clean-up reactions} $v + X \xrightarrow{R} \phi$, and \emph{covering reactions} $v + e \xrightarrow{R} v$. The initial and target configurations remain the same. 
\end{proof}

\subsection{iCRNs with General rules}
\begin{theorem} \label{thm:iCRN (1,1) Pspace-hardness}
The reachability problem in iCRNs with $(1,1)$ rules is PSPACE-complete.
\end{theorem}

\begin{proof}
PSPACE membership follows from Lemma \ref{lem:pspace}. We now show hardness. Given a Turing machine $\mathcal{M}= (Q,\Sigma,\Gamma,\delta,q_s,q_a,q_r)$ with bounded tape length $n$, we create an iCRN $\crn{C_{IC}} = (\Lambda,\Gamma_C)$. The set $\Lambda$ consists of the helper species $D$, $H$, and $E$, the tape symbols $\{x_i, 0_i, 1_i\}$ for each tape cell $i \in [1,n]$, and the following sets of species: $\mathcal{Q} = \{Q_y | q_y \in Q\}$, $\mathcal{Q}^0 = \{\, Q_y0_i \,\}_{q_y \in Q,\; i \in [1,n]}$ and $\mathcal{Q}^1 = \{\, Q_y1_i \,\}_{q_y \in Q,\; i \in [1,n]}$. Table \ref{tab:?} shows the reactions created.

The initial configuration $I$ of the system consists of one instance of $0_i$ or $1_i$ for each tape cell $i$, according to the initial tape of $\mathcal{M}$. Given the start state $q_s$, we also add a single copy of $Q_s$. If the tape head starts at cell $i$, we include one instance of each species $x_j$ for all $j \neq i$, i.e., $\{x_j\}_{j \neq i}$. Finally, we include one copy each of the helper species $D$ and $H$. The target configuration $T$ is volume $I$  copies of $E$.

\begin{table}[t]
\centering
        \begin{tabular}{|l|l|l|l|}\hline
            \textbf{No.} & \textbf{Notation} & \textbf{Reaction} & \textbf{Intuition} \\ \hline
            1.1 & $\mathcal{Q}_H: {Q_a, Q_r}$
            & $D\xrightarrow{\mathcal{Q}_H,\overline{\mathcal{Q}_y}, x_i, 0_i}Q_y1_i$
            & Read the tape can\\
    
            1.2 & $\overline{\mathcal{Q}_y}: \mathcal{Q}\setminus Q_y$
            & $D\xrightarrow{\mathcal{Q}_H,\overline{\mathcal{Q}_y}, x_i,1_i}Q_y0_i$ 
            & combine with state,  \\
    
            1.3 & $\overline{\mathcal{Q}_H} : \mathcal{Q} \setminus \mathcal{Q}_H$
            & $D\xrightarrow{\overline{\mathcal{Q}_H}}E$ 
            & unless state is halt. \\ \hline
    
            \multirow{3}{*}{2} & \multirow{2}{*}{$\overline{\mathcal{Q}^V}_{\!\{{y,i}\}}: \{\mathcal{Q}^0\cup \mathcal{Q}^1\} \setminus Q_yV_i$} 
            & \multirow{3}{*}{$Q_y\xrightarrow{D,E,\overline{\mathcal{Q}^V}_{\!\{{y,i\}}}}Q_{zi}$}
            & Prepare to change \\

            & \multirow{2}{*}{$z: \delta_{state}(y,V)$}
            & 
            & into the appropriate  \\ 

            & 
            & 
            &  $z$ according to $\delta$. \\ \hline
    
            3.1 &  
            & $1_i\xrightarrow{D,E,\mathcal{Q},x_i} x_i$ 
            & Mark read as done by  \\
    
            3.2 &  
            & $0_i\xrightarrow{D,E,\mathcal{Q}, x_i}x_i$ 
            &  producing $x_i$. \\ \hline
    
            \multirow{2}{*}{4}  &  
            & \multirow{2}{*}{$Q_{zi}\xrightarrow{1_i,0_i}Q_z$} 
            & Finish changing state \\ 

            &&& after producing $x_i$. \\ \hline
    
            \multirow{2}{*}{5} & $b:\delta_{dir}(y,V)$
            &   \multirow{2}{*}{$x_{i+b}\xrightarrow{E,1_i,0_i,\overline{\mathcal{Q}^V}_{\{y,i\}},Q_{zi}}{c_i}$} 
            & Move+write based on \\  
            
            &  $c:\delta_{write}(y,V)$
            & 
            & state and value read $V$. \\ \hline
    
            6 &
            & $Q_yV_i\xrightarrow{x_{i+b}}D$ 
            & Once moved, replace $D$. \\ \hline
    
            7.1 &  ${\mathcal{Q'}} : \{\mathcal{Q}^1 \cup \mathcal{Q}^0\}$
            & $H\xrightarrow{D,{\mathcal{Q}'}} E$ 
            & If $\mathcal{M}$ halts, delete $H$, \\ 
    
            7.2 & $S: \Lambda \setminus \{E,Q_r\}$ 
            & $S\xrightarrow{H}E$ 
            & and go to term config. \\ \hline
        \end{tabular}
    
    \caption{Reaction set for the $(1,1)$ iCRN construction.}\label{tab:?}
    \vspace{-.5cm}
\end{table}

\para{Forward Direction:} Assume we are given a Turing Machine that will reach a halt state by applying a series of transitions from $\delta$, each according to the tape value at head and the state. The corresponding iCRN $\crn{C_{IC}}$ will likewise reach target configuration $T$ by applying a series of reactions representing those transitions, according to the current configuration $C$. When a single $D$ is present in $C$, $C$ maps to the Turing Machine. $Q_y$ indicates $\mathcal{M}$ is in state $q_y\in Q$. The Tape head is at tape cell $j$ where $x_j$ is absent. The data for tape cell$_i$ is saved in the presence of a copy of $1_i$, or one of $0_i$. All reactions shown in Table \ref{tab:?} except 1.3, 7.1 and 7.2 handle normal transitions, updating the given configuration to map to the result of the next transition. In this order, the reactions: (1) read the tape, (2) determine the next state, (3) mark the tape head ready to move, (4) change the state, (5) write and move the tape head, and (6) indicate ready for next transition. If the current state is a halt state, reaction 1.3 is applied, than (7) is repeatedly applied until we reach a terminal configuration. The final terminal configuration will only be $T$ if the TM's tape and state input deterministically transitions to $q_a$ according to $\delta$. 

\para{Reverse Direction:} 

To reach the target configuration $T$ of iCRN $\crn{C_{IC}}$ which was designed to simulate bounded length $N$ TM $A$, reaction 7.2 must run $2N$ times. This can only happen if $Q_a$ was present, preventing reactions 1.1 and 1.2, and not being $Q_r$, which would also prevent the reactions, but could not be removed by 7.2. For $Q_a$ to be present, it either was in $I$, or was produced by a type 4 reaction. To have $Q_a$ produced by a type 4 reaction, the most recent type 2 reaction must have produced $Q_{ai}$. For any reaction of type 2 to be applicable, D cannot be present, so before this reaction could occur, a reaction of type 1.1 or 1.2 must have been run. It produced a copy of $Q_yV_i$ where $Q_y$ represents the state the system was in before $Q_a$, and $V$ is the symbol at our tape head. $Q_yV_i$ prevents all type 2 reactions except $\delta_{state}(y,V)$, so we could only have reached $Q_a$ if $\delta_state(y,V)$ = $q_a$. For our configuration to have included a $1_i$ or $0_i$ and $Q_y$ which would lead to $Q_a$, they either must have been in $I$, or were produced by repetitions of the same process. This process also includes one reaction of type 3 and one of type 5, which together update the tape. Thus we could only have reached the target configuration if the initial tape values and State deterministically Halt in an accept state.  
\end{proof}

\begin{theorem}
The reachability problem in iCRN with $(1,1)$ rules and at most 1 inhibitor per rule is NP-Hard.    
\end{theorem}
\begin{proof}
We show this by a reduction from 3SAT.  Consider a 3SAT instance of $n$ variables $X_1, \ldots X_n$ and $m$ clauses $C_1,\ldots C_m$, and assume no clause contains both $x_i$ and $\overline{x_i}$. From this instance, we create the following $(1,1)$ iCRN system and start and end configurations for the Reachability problem.

For each variable $x_i\in X_n$, we create the species $y_i$, $x_i$, $T_i$, and $F_i$. For each clause $c_j\in C_m$, we create the species $c^1_j$,$c^2_j$,$c^3_j$, $d^1_j$,$d^2_j$,$d^3_j$, and $S_j$.
Additionally, for each variable, we create the reactions $y_i \rightarrow T_i$, $y_i \rightarrow F_i$, $x_i \xrightarrow{x_{i-1}} T_i$ (no inhibitor for $i=1$), and $x_i \xrightarrow{x_{i-1}} F_i$ (no inhibitor for $i=1$). For each clause, we create the reactions $c^k_j \xrightarrow{x_n} d^k_j$, either $d^k_j \xrightarrow{T_i} S_j$ if $\overline{x_i}$ satisfies the $k^{th}$ literal of clause $C_j$ or $d^k_j \xrightarrow{F_i} S_j$ if $x_i$ satisfies the $k^{th}$ literal of clause $C_j$, and $d^k_j \rightarrow W$.
Finally, let the initial configuration be one copy for each of $y_i$, $x_i$, and $c^1_j$, $c^2_j$, $c^3_j$, and the target configuration be one copy for each of $T_i$, $F_i$, $S_j$, and $2m$ copies of $W$.

\para{Forward Direction:} Suppose the 3SAT formula is satisfiable.  We can then reach the destination configuration by first transitioning each $x_i$ species into the appropriate true/false species $T_i$/$F_i$ that matches the satisfying assignment for the 3SAT formula.  Next, convert exactly one of the three $c_j$ species from each clause into $S_j$ (based on which literal was satisfied by the satisfying variable assignment), and the remaining two into copies of species $W$.  Finally, convert each $y_i$ into the opposite truth value taken on by $x_i$ (to ensure there is exactly one $T_i$ and $F_i$ for each $i$).

\para{Reverse Direction:}  Suppose we can reach the destination configuration, which specifically requires creating each species $S_j$.  Each $S_j$ can only be created by the absence of some $F_i$ where the clause contains literal $X_i$, or the absence of a $T_i$ where the clause contains $\bar{X_i}$.  Since the $c_j$ to $d_j$ transformation must precede the creation of $S_j$, we know that each $x_i$ must be converted to a $T_i$ or $F_i$ before the $S_j$'s can be created.  Therefore, the only possible way to create each $S_j$ is for each $x_i$ to select a truth assignment that corresponds to a satisfying assignment of the given 3SAT formula.
\end{proof}

\begin{observation}\label{obs:(2,2)-icrn}\vspace{-.2cm}
    The reachability problem for iCRNs with rule size $(2,2)$ is PSPACE-complete.
\end{observation}
\begin{proof}
    The proof is analogous to the proof of Observation~\ref{obs:(2,2)-Picrn}.
\end{proof}

\section{Conclusion}\label{sec: conclusion}

This paper provides a thorough treatment of reachability in restricted CRNs with inhibition. Our results, summarized in Table 1, show a clear boundary between tractable / intractable problems for CRNs with prioritized and unprioritized inhibition.  For Priority iCRNs, the reachability problem remains solvable in polynomial time for most unimolecular, bimolecular, and mostly catalytic void rules, but becomes NP-complete for systems with a combination of $(2,0)$ and $(2,1)$ void rules.  For general iCRNs, almost every case becomes NP-complete (with the exception of $(1,0)$ rules), and we provide FPT algorithms for the $(2,0)$, $(2,1)$, and $(k,k-1)$ cases.  As for volume-preserving reactions, we show that reachability in systems of $(1,1)$ rules---which is solvable in polynomial time for basic CRNs---drastically increases to PSPACE-complete for general iCRNs.

We conclude by outlining some natural extensions and future directions. First, a number of complexity gaps remain in our classification of the reachability problem in inhibitory systems. These gaps, posed as open problems in Appendix~\ref{subsec:gaps}, represent the remaining questions needed to complete the problem landscape in Table~\ref{tab:results}.

Second, we identify an intriguing line of work that focuses on extending CRNs with a set of global states. A \emph{CRN with States} can be viewed as a structured form of inhibition, where certain reactions are disabled based on the current global state. We formalize this model in Appendix~\ref{subsec:crnss}, and show that in this model the reachability problem is NP-hard even for the smallest deletion-only rules, and PSPACE-complete for the smallest volume-preserving rules.



\newpage
\bibliography{crns}
\newpage
\appendix
\section{Proofs for NP-hardness in Void Priority-iCRNs}\label{sec:np-hard-picrn}
\picrnVC*

\begin{proof}
    From Lemma \ref{thm:NP-voidiCRN}, we have that the reachability problem for Priority iCRNs with purely void rules is in NP. We now show hardness for systems with rules of size $(2,0)$ and $(2,1)$. For this we reduce from Vertex Cover problem. Given an instance of VC $\langle G,k \rangle$ where $k>0$, we construct our Priority iCRN $\crn{C_{PI}}$ as follows. We also add two new species to the $\crn{C_{PI}}$ system. In the initial configuration, we add $|V|-k$ copies of a species $R$, used to eliminate $|V|-k$ vertices, and $k$ copies of a species $X$, used to remove the remaining $k$ vertices. Denote this configuration the starting configuration $\config{C_s}$. Then our ordered species set is given as
    \[
        \species = \langle R, v_1,\ldots,v_{\size{V}}, e_1,\ldots,e_{\size{E}, X}\rangle
    \]

    We now create the reactions. For all $v\in V$, we create the \emph{assignment reaction} $v+R\rightarrow\emptyset$ to create the vertex cover and the \emph{clean-up reaction} $v+X\xrightarrow{1}\emptyset$ to remove the remaining $v$ copies. For all edges $e\in E$, for all vertices $v\in V$ where $v$ is incident to $e$, we create the \emph{covering reaction} $v+e\xrightarrow{1}v$, which represents edge $e$ being covered by vertex $v$. Here each of the covering reactions have priority $1$.

    The target configuration is the empty configuration $\config{0}$.

    \para{Forward Direction.}
    Assume there exists a vertex cover of size $k$ in $G$. We can then apply a sequence of assignment reactions to $\config{C_s}$ that deletes each copy of $v$ whose respective vertex is \emph{not} in the vertex cover, as well as all $R$ copies. Denote this new configuration $\config{C_s}'$. With no copies of $R$ present in $\config{C_s}'$, the covering reactions can be fired now. Because the remaining $v$ copies represent the vertex cover, by the construction of $\crn{C_{PI}}$, all $e$ copies can be deleted by the covering reactions. Then the $k$ copies of $v$ and $X$ remaining can be deleted with the clean-up reactions, resulting in a final configuration of $\config{0}=\config{C_t}$
    
    \para{Reverse Direction.}
    Assume there exists a sequence of applicable reactions that transitions $\config{C_s}$ to $\config{C_t}$. The only method to delete all copies of $R$ in $\config{C_s}$ is to apply a sequence of assignment reactions that consumes all $R$ copies and $|V|-k$ copies of $v$. Denote the resulting configuration $\config{C_s}'$.
    
    Assume there exists a sequence of covering reactions that can delete all copies of $e$ from $\config{C_s}'$. By the construction of $\crn{C_{PI}}$, because each covering reaction corresponds to an edge of $G$, deleting all copies of $e$ implies that the remaining copies of $v$ correspond to a vertex cover of size $k$ in $G$. If no sequence of covering reactions exists that can delete all $e$ copies, then the empty configuration cannot be reached.
\end{proof}

\orderediCRN*
\begin{proof}
    Given an iCRN $\crn{C_{IC}}=((\species,\reactions),\mathcal{I})$, start configuration $\config{C_s}$ and an ordering on inhibitors $\langle i_1,\ldots,i_c\rangle$, we construct a Priority iCRN $\crn{C_{PI}} = ((\species', \reactions'), \mathcal{P})$ and start configuration $\config{C_s'}$ as follows.
    The ordered set of species in the Priority iCRN is $\species' = \langle i_1,\ldots,i_c,j_1,\ldots,j_{\size{\species}-c}\rangle$ where $\{j_1,\ldots ,j_{\size{\species}-c}\}=\species\setminus\{i_1,\ldots,i_c\}$. Thus, the inhibitor species occupy the first $c$ positions.

    We then define a mapping $index:\species'\rightarrow[1,\size{\species'}]$ where $index(\lambda) = k \iff \species'[k]=\lambda$. Each reaction $\reaction = (\reactants,\products)\in\reactions$ is transformed into $\reaction' = (\reactants',\products')\in\reactions'$ by permuting the coordinates of both $\reactants$ and $\products$ according to $index$, and the initial configuration vector $\config{C_s'}$ is obtained by rearranging $\config{C_s}$ in the same manner. Next, for each reaction $\reaction'\in\reactions'$, corresponding to $\reaction \in \reactions$, we define $\mathcal{P}(\reaction') = \max_{\lambda\in\mathcal{I}(\reaction)}index(\lambda)$. Finally, we remove any reactions that are self-inhibiting, i.e., the priority of the reaction is greater or equal to the index of at least one of the reactants.
\end{proof}

\section{Proofs for NP-hardness in Void iCRNs}\label{sec:np-hard-appendix}
We now show that the reachability problem for iCRNs with only $(2,0)$-size rules is NP-hard by a reduction from the Hamiltonian Path (HAMPATH) problem.
We can transform any instance of HAMPATH $\langle G, s,t \rangle$ into an iCRN $\crn{C_{IC}}$ in which each species represents a vertex in $G$ and each void reaction represents the path visiting these vertices. Then a Hamiltonian path from $s$ to $t$ in $G$ only exists iff the system can delete all but the $t_c$ copy. The complete construction of $\crn{C_{IC}}$ is described as follows:

\para{Species.}
For each vertex $v \in V$, we create the species $v$ and $v_c$. The presence of a copy of $v$ in a configuration of $\crn{C_{IC}}$ indicates its respective vertex has yet to be visited by the path, while $v_c$ can be thought of as a token from $v$ that the path must use to visit a new vertex connected from $v$. We also add an extra species $s'$ to start the path at vertex $s$.

\para{Reactions.}
For each directed edge $\{u,v\} \in E$, we create the \emph{choosing reaction} $u_c+v\xrightarrow{u}\emptyset$, which represents the path choosing a new vertex $v$ to visit from $u$. As species $u$ inhibits the reaction, it can only be applied if a previously-applied choosing reaction consumed $u$ (i.e. $s_c+u\xrightarrow{s}\emptyset$). Because each choosing reaction is initially inhibited by some $v$ species, we also create the \emph{starting reaction} $s'+s\rightarrow\emptyset$ to be the first reaction applied in the configuration, as well as ensure the path always starts at vertex $s$.

An example construction of $\crn{C_{IC}}$ is illustrated across Figures \ref{subfig:(2,0)-icrn-reduction-1} and \ref{subfig:(2,0)-icrn-reduction-2}.

\icrnhampath*

\begin{proof}
    Lemma \ref{thm:NP-voidiCRN} shows reachability in void iCRNs is in NP. To show hardness, we reduce from the Hamiltonian path problem. Given a HAMPATH instance $\langle G,s,t \rangle$, we transform it into an instance of reachability with an iCRN $\crn{C_{IC}}$, following the construction above, an initial configuration $\config{C_s}$ of a single copy of each species of $\crn{C_{IC}}$ (See Figure \ref{subfig:(2,0)-icrn-reduction-3}), and a target configuration $\config{C_t}$ of only one copy of $t_c$ (or $\single{t}_c$).

    \para{Forward Direction.}
    Assume there exists a Hamiltonian path from $s$ to $t$ in $G$. For $\crn{C_{IC}}$, the only applicable rule in $\config{C_s}$ is the starting reaction, which deletes the $s$ and $s'$ copy. Then, by following the edge sequence of the Hamiltonian path, we can apply the choosing reaction for each edge to delete the involved $u_c$ and $v$ copies. No choosing reaction $r$ is inhibited prior to application because the inhibitor species of $r$ must be consumed by the previous choosing reaction. The last choosing reaction leaves only one $t_c$ copy as the final configuration, which matches $\config{C_t}$.

    \para{Reverse Direction.}
    Assume there exists a sequence of applicable reactions from $\crn{C_{IC}}$ that transitions the initial configuration $\config{C_s}$ to $\config{C_t}=\single{t}_c$. First, the only applicable reaction in $\config{C_s}$ is the starting reaction, which deletes the $s$ and $s'$ copies and starts the path at $s$. Denote the resulting configuration $\config{C_s}'$. Removing the $s$ copy allows all choosing reactions inhibited by the $s$ species to be applied to $\config{C_s}'$, which corresponds to the path selecting a new vertex to visit from $s$. Because only one copy of $s_c$ is present, only one of these choosing reactions may be applied; thus, the path can only choose one vertex to visit at a time. This also makes it impossible for the path to backtrack to $s$ and select a new vertex to visit. After applying a choosing reaction and deleting the copy of $s_c$ and some species $u$, the path arrives at the new vertex $u$, which follows the same case regarding choices and backtracking as from vertex $s$.

    Assume a sequence of choosing reactions can be applied in $\config{C_s}'$ that deletes all copies except the $t_c$ copy. By the construction of $\crn{C_{IC}}$, each choosing reaction in the sequence must correspond to an edge in $G$ visited by the path. Because all but the $t_c$ species were removed, this implies the existence of a Hamiltonian path from $s$ to $t$ in $G$. If a vertex cannot be reached by a path, then its respective $v$ and $v_c$ copies will never be removed. If a Hamiltonian path exists, but $t$ is not the last vertex visited, then the $t_c$ copy will be deleted by a choosing reaction. Therefore, if $\config{C_t}$ can be reached from $\config{C_s}$, then a Hamiltonian path from $s$ to $t$ must exist in $G$.
\end{proof}

We now show that reachability in iCRN is still NP-complete even for $(2,1)$-size reactions. Here, we reduce from the classic 3SAT problem. We transform any 3SAT instance $\langle\Phi\rangle$ into a new iCRN $\crn{C_{IC}}$ where the variables and clauses are represented by species and variable assignments, and clause satisfactions are represented by reactions. The formula $\Phi$ can only be satisfied iff $\crn{C_{IC}}$ can consume all species copies except for one. The complete construction of $\crn{C_{IC}}$ is provided below.

\para{Species.} For each variable $x_i\in X_n$, we create the species $T_i$ and $F_i$, which represent an assignment of true or false to $x_i$, respectively. Additionally, for each clause $c_j\in C_m$, we create the species $c_j$. The presence of a copy of $c_j$ in a configuration indicates its respective clause is not yet satisfied by some assigned variable. Finally, we create the species $X$ for clean-up purposes.

\para{Reactions.} For each variable $x_i$, we create the \emph{assignment reactions} $T_i + F_i \rightarrow T_i$ and $T_i + F_i \rightarrow F_i$. The remaining copy represents assigning $x_i$ the respective truth value of that copy. For each clause $c_j$ and a variable of $c_j$ $x_k$, we create the clause reaction $C_i + T_k \xrightarrow{F_k} T_k$ if assigning true to $x_k$ satisfies $c_j$, or $C_i + F_k \xrightarrow{T_k} F_k$ if the false assignment satisfies instead. Finally, for each variable, we create the \emph{clean-up reaction} $T_i + X \rightarrow X$ and $F_i + X \rightarrow X$.

An example construction of $\crn{C_{IC}}$ is illustrated across Figures \ref{subfig:(2,1)-icrn-reduction-1} and \ref{subfig:(2,1)-icrn-reduction-2}.

\icrnsat*
\begin{proof}
From Lemma \ref{thm:NP-voidiCRN}, we have that the reachability problem in void iCRNs is in NP. To show hardness, we reduce from the 3-satisfiability problem. Given an instance of 3SAT $\langle\Phi\rangle$, we create an iCRN $\crn{C_{IC}}$ following the reduction above. Let the initial configuration $\config{C_s}$ be a single copy for each species, and the target configuration $\config{C_t}$ be only the copy of $X$.

\para{Forward Direction.}
Assume there exists an assignment of variables that satisfies $\Phi$. Then we can apply a sequence of assignment reactions to $\config{C_s}$ in which the remaining single copy of $T_i$ or $F_i$ for each variable $x_i$ matches its truth value in the satisfying assignment. Because the variable assignment satisfies $\Phi$, all $C_j$ copies can be deleted by clause reactions. Then the copy of $X$ consumes all remaining $T_i/F_i$ copies with the clean-up reactions, leaving the single copy of $X$ as the final configuration that matches $\config{C_t}$.

\para{Reverse Direction.}
Assume there exists a sequence of applicable reactions from $\crn{C_{IC}}$ that transitions the initial configuration $\config{C_s}$ to $\config{C_t}=\single{X}$. First, an assignment reaction is required to delete one of the single copies of $T_i$ and $F_i$ for each variable. Then a sequence of assignment reactions may be applied to $\config{C_s}$ that results in one single $T_i/F_i$ copy per variable. Denote this configuration $\config{C_s}'$. Although we can now immediately use the $X$ copy and the clean-up reactions to delete the remaining $T_i/F_i$ copies, this leaves the copies of $C_j$ present; thus, $\config{C_t}$ cannot be reached in this case.

Assume there exists a sequence of clause reactions that deletes all $C_j$ copies in $\config{C_s}'$. By the construction of $\crn{C_{IC}}$, because each clause reaction corresponds to a clause $c_j$ being satisfied by some assigned literal, the deletion of all $c_j$ species implies that the assignment represented by $\config{C_s}'$ satisfies $\Phi$. Note that a clause reaction for a specific variable $x_i$ can only be fired once the assignment reaction for that variable is fired (otherwise it is inhibited), ensuring that $\crn{C_{IC}}$ must assign variables first before satisfying clauses and preventing both of the $T_i$ and $F_i$ copies from the same variable from deleting $C_j$ copies together. If no sequence of clause reaction exists to delete all $C_j$ copies, then this will leave at least one $C_j$ copy present, preventing $\config{C_t}=\single{X}$ from being reached. Therefore, if $\config{C_t}$ can be reached from $\config{C_s}$, then there must exist an assignment of variables that satisfies $\Phi$.
\end{proof}

\section{Extensions and Open Problems}
\subsection{Resolving Complexity Gaps}\label{subsec:gaps}

While many of our classifications are tight, several important gaps need to be resolved in order to complete Table~\ref{tab:results}.  These gaps are natural targets for future work.

\para{Open Problem.}\label{open: w1-hardness} In the case of inhibitory systems with void rules of size $(2,0)$ and $(2,1)$, we show NP-completeness even with systems limited to 1 inhibitor/priority. Does there exist some other parameter, such as the number of inhibited rules, that lends itself to an FPT algorithm? Or the contrary: is the problem W[1]-hard?

\para{Open Problem.}\label{open: fpt-improvements} Can we improve the complexity of the FPT algorithms for iCRNs? Such an improvement might require a new method to reduce the number of cases instead of trying all the permutations of inhibitors.

\para{Open Problem.}\label{open: (1,1)} What is the complexity of $(1,1)$ Priority iCRNs? For basic CRNs, this problem is solvable in polynomial time, but jumps all the way to PSPACE-complete when general inhibition is allowed.  How does prioritized inhibition change the problem?  Does it match basic CRNs or iCRNs? Or does it fall somewhere in between, such as NP-complete?

\subsection{CRN with States}\label{subsec:crnss}
\emph{Vector Addition System with States} (VASS) is an extension of Vector Addition Systems that maintains a global state variable \cite{HOPCROFT1979135} that dictates which vectors can be added. We can extend CRNs in the same way by considering a set of states $\states$.
Formally, we say a \emph{CRN with States} $\mathcal{C_{S}} = ((\species, \reactionsS), \states)$ is defined as a CRN where the set of reactions is given by $\reactionsS \subseteq \states\times\reactions\times\states$.

\begin{definition}[CRN States Dynamics]
    For a CRN with states $((\species, \reactionsS), \states)$ and state configurations $A = (q_A, \config{C_A})$ and $B = (q_B, \config{C_B})$, we say that $A \rightarrow_{\crn{C_{S}}}^{(\species, \reactions)} B$ if there exists a rule $(q_i, \reaction, q_f) \in \reactionsS$ where $\reaction = (\reactants, \products)$ such that $q_i = q_A$, $q_f=q_B$, $\reactants \leq \config{C_A}$, and $\config{C_A} -\reactants+\products=\config{C_B}$.
\end{definition}

CRNs with states can be thought of as a form of inhibitory CRN: the states inhibit reactions that would otherwise be possible based on configuration counts alone.  Thus, it warrants investigating the complexity of reachability in this model.
As initial preliminary work in this direction, we show that the reachability problem for CRNs with states is NP-hard even for void rules of size $(1,0)$ by reducing from HAMPATH, and it becomes PSPACE-complete for systems with size $(1,1)$ rules by reducing from zero-player motion planning \cite{demaine2023pspace}. 
These results are contrasted with reachability for basic CRNs in Table~{\ref{tab:results with states}}.

This section explores the reachability problem in the CRN model with States. We start by considering void rules of different sizes, showing that the problem is hard even for void rules of size $(1,0)$. We then consider systems with rules of size $(1,1)$, showing that the problem is PSPACE-complete.

\begin{table}[t]
    \vspace{-.2cm}
    \centering
    \begin{tabular}{|c|c|c|c|c|c}
    \hline
     \multicolumn{5}{|c|}{\textbf{Reachability with States}} \\
     \hline \hline
     \textbf{Rule Size} & \textbf{\hyperref[subsec:crn]{\textbf{CRN}}} & \textbf{Ref.} & \textbf{\hyperref[subsec:crnss]{\textbf{CRN with States}}} & \textbf{Ref.}\\
     \hline
     (1,0) & P & - & \cellcolor{blue!15} \hyperref[open: NP membership]{NP-hard}  & Thm. \ref{thm:(1, 0) picrns} \\
     \hline
     \multicolumn{5}{|c|}{\textbf{General Systems}}\\
     
     \hline
     (1,1) & P, NL-hard & \cite{AFG:2025:RRC} & PSPACE-c & Thm. \ref{thm:(1, 1) picrns}\\
     \hline
     General & Ack-c& \cite{czerwinski2022reachability, leroux2022reachability} & Ack-c & Obs. \ref{obs:general-crnss}\\
     \hline
    \end{tabular}
    \caption{Reachability in CRNs with states. The cell highlighted in \colorbox{blue!15}{blue} links to an open direction discussed in Section~\ref{subsec:crnss}.}
    \label{tab:results with states}
    \vspace{-.5cm}
\end{table}

\begin{restatable}{theorem}{crnssVoid}\label{thm:(1, 0) picrns}
    The reachability problem in CRNs with states, with only void rules of size $(1,0)$, is NP-hard.
\end{restatable}
\begin{proof}
    Given a Ham-Path $H$ with $n$ nodes $N_1, N_2, \ldots N_n$, edges $E_1, E_2, \ldots E_i$, and starting node $N_s$, we construct a Priority iCRN $C_{pi}$ with $n$ states, $i$ reactions, $n$ species and $n$ volume as follows. The species list of $C_{pi}$ is $N_1', N_2', \ldots N_n$, and the state list is $q_1, q_2, \ldots q_n$. The initial configuration is 1 count of every species in the species list of $C_{pi}$, with initial state $q_s$. The rules allowed in each state $q_k$ are, for each edge $E_j$ connecting $N_k$ to $N_l$, $(q_k, N_k'\rightarrow\emptyset, q_l)$. The target configuration is $\config{\emptyset}$ in state $q_s$. 

\end{proof}

\begin{corollary} \label{cor:picrns}
    The reachability problem in CRNs with states, with only void rules of bigger sizes, is NP-hard.
\end{corollary}

We now consider non-void rules, showing that even for systems with size $(1,1)$ rules, the problem is PSPACE-Complete through a reduction from zero-player motion planning. In this problem, we are given an initial configuration of gadgets, a starting location, and a target location, and ask whether an agent can traverse through the gadgets from the starting location to reach the target location. By showing how to simulate a locking 2-toggle and a rotate gadget with $(1,1)$ rules, we achieve PSPACE-hardness \cite{demaine2023pspace}.

\begin{figure}
    \centering
    \begin{subfigure}[b]{0.3\textwidth}
        \centering
        \includegraphics[width=.8\textwidth]{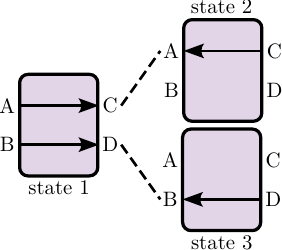}
        \caption{}
        \label{subfig:Locking-2-Toggle}
    \end{subfigure}
    \hfill
    \begin{subfigure}[b]{0.5\textwidth}
        \centering
        \includegraphics[width=1.\textwidth]{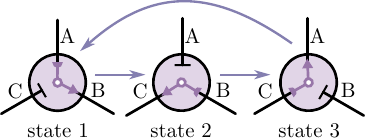}
        \caption{}
        \label{subfig:Rotating-Gadget}
    \end{subfigure}
    \hfill
    \caption{(a) A locking 2-toggle with 3 possible states. (b) A rotating turnstile gadget.}
    \label{fig:Gadgets}
\end{figure}

\textbf{Wires.} Each wire is given a label as shown in Figure \ref{subfig:Locking-2-Toggle} and is represented by 2 states $\overrightarrow{q_i}$ and $\overleftarrow{q_i}$. The orientation of the arrow denotes what direction the agent is traveling in. 

\textbf{Locking 2-Toggle (L2T).} Each L2T is enumerated and represented as species $G_i^1, G_i^2, G_i^3$ for each possible state. All possible interactions by an agent with an L2T are represented by the set of rules. In Figure \ref{subfig:Locking-2-Toggle}, a valid traversal through $G_i$ in state $1$ would be represented as $(\overrightarrow{q_a}, G_i^1 \rightarrow G_i^2, \overrightarrow{q_c})$. An example where the agent `bounces off' $G_i$ would be represented as $(\overleftarrow{q_c}, G_i^1 \rightarrow G_i^1, \overrightarrow{q_c})$.

\textbf{Rotating Gadget.} Each rotating gadget is also enumerated and represented as a single species $R_i$. All possible interactions by an agent with a rotating gadget are represented by the set of rules. In Figure \ref{subfig:Rotating-Gadget}, a valid example traversal through $R_i$ would be represented as $(\overrightarrow{q_a}, R_i^1 \rightarrow R_i^2, \overrightarrow{q_b})$.

\begin{restatable}{theorem}{crnss}\label{thm:(1, 1) picrns}
    The reachability problem in CRNs with states, with only rules of size $(1,1)$, is PSPACE-complete.
\end{restatable}
\begin{proof}
First, we observe that the problem is contained within PSPACE:  We can non-deterministically select a rule that is ready to apply, and it takes polynomial space to maintain the current system configuration, which is just the number of copies of each species. This gives a nondeterministic polynomial space simulation for the CRN system that as the volume does not change over time. By the classical result NSPACE=PSPACE, we have membership in PSPACE.

We now show hardness through a reduction from zero-player motion planning. Given an instance of a motion planning problem with $n$ Locking 2-Toggles and $m$ rotate gadgets, create a CRN with states system as described above, with $\overrightarrow{q_0}$ denoting the initial position of the agent and $q_t$ denoting the target position. We add rules for each gadget such that once the system reaches state $q_t$, then all species transition to a dummy species $F$. If the system can reach the final configuration of $n + m$ copies of species $F$ in state $q_t$, then there exists a path from the initial location to the target location in the motion planning problem.
\end{proof}

\begin{restatable}{observation}{crnssGeneral}\label{obs:general-crnss}
    The reachability problem in CRNs with states, with general rule sizes, is Ackermann-complete since VAS and VASS are equivalent.
\end{restatable}

\para{Open Problem.}\label{open: NP membership} We show that CRNs with states, when limited to void rules of size $(1,0)$, are NP-hard. Is the 
problem contained within NP for such systems, or is the problem PSPACE-hard?

\end{document}